\documentclass[final,hidelinks,oneeqnum,onefignum,onetabnum,onealgnum,onethmnum]{siamart220329}

\usepackage[utf8]{inputenc}
\usepackage[english]{babel}

\usepackage[linesnumbered,lined,ruled,vlined,algo2e]{algorithm2e}
\SetKw{KwForIn}{in}
\SetKwInput{KwInputs}{Inputs}
\SetKwInput{KwOutputs}{Outputs}
\SetKwInput{KwInParallel}{(in parallel)}
\SetKwProg{Func}{Function}{:}{end}
\SetKwFunction{FnMQuad}{moment\_quadrature}

\usepackage{amsmath,amsfonts,amssymb,bm}
\usepackage{mathtools}

\usepackage{graphicx}
\usepackage{hyperref}
\hypersetup{colorlinks=false, hidelinks}
\usepackage{booktabs}
\usepackage{url}

\newcommand{\vscriptalign}[1]{#1\vphantom{\meaning#1}}

\newcommand{\cu}[1]{
	\ifcat\noexpand#1\relax
	\bm{#1}
	\else
	\mathbf{#1}
	\fi
}

\newcommand{\diff}{\mathop{}\!\mathrm{d}}

\newcommand{\imag}{\mathrm{i}}

\newcommand{\cond}{{\;|\;}}

\let\sup\relax

\let\lim\relax
\DeclareMathOperator*{\sup}{sup\,}  
\DeclareMathOperator*{\lim}{lim\,}  

\newcommand{\expecsym}{\operatorname{\mathbb{E}}}     
\newcommand{\covsym}{\operatorname{Cov}}     
\newcommand{\varrsym}{\operatorname{Var}}     
\newcommand{\diagsym}{\operatorname{diag}}     
\newcommand{\tracesym}{\operatorname{tr}}           

\let\expec\relax
\let\cov\relax
\let\varr\relax
\let\diag\relax
\let\trace\relax

\makeatletter
\newcommand{\expec}{\@ifstar{\@expecauto}{\@expecnoauto}}
\newcommand{\@expecauto}[1]{\expecsym \left[ #1 \right]}
\newcommand{\@expecnoauto}[1]{\expecsym [#1]}

\newcommand{\cov}{\@ifstar{\@covauto}{\@covnoauto}}
\newcommand{\@covauto}[1]{\covsym \left[ #1 \right]}
\newcommand{\@covnoauto}[1]{\covsym [#1]}

\newcommand{\varr}{\@ifstar{\@varrauto}{\@varrnoauto}}
\newcommand{\@varrauto}[1]{\varrsym \left[ #1 \right]}
\newcommand{\@varrnoauto}[1]{\varrsym [#1]}

\newcommand{\diag}{\@ifstar{\@diagauto}{\@diagnoauto}}
\newcommand{\@diagauto}[1]{\diagsym \left( #1 \right)}
\newcommand{\@diagnoauto}[1]{\diagsym (#1)}

\newcommand{\trace}{\@ifstar{\@traceauto}{\@tracenoauto}}
\newcommand{\@traceauto}[1]{\tracesym \left( #1 \right)}
\newcommand{\@tracenoauto}[1]{\tracesym (#1)}

\newcommand{\traceBig}[1]{\tracesym \Bigl( #1 \Bigr)}

\makeatother

\newcommand*{\trans}{{\mkern-1.5mu\mathsf{T}}}

\newcommand*{\R}{\mathbb{R}} 

\newcommand*{\PP}{\mathbb{P}} 

\let\norm\relax
\DeclarePairedDelimiter{\normbracket}{\lVert}{\rVert}
\newcommand{\norm}{\normbracket}

\newcommand{\normBigg}[1]{\Biggl\lVert #1 \Biggr\rVert}

\let\innerp\relax
\DeclarePairedDelimiter{\innerpbracket}{\langle}{\rangle}
\newcommand{\innerp}{\innerpbracket}

\let\abs\relax
\DeclarePairedDelimiter{\absbracket}{\lvert}{\rvert}
\newcommand{\abs}{\absbracket}
\newcommand{\absbig}[1]{\bigl \lvert #1 \bigr\rvert}

\newcommand{\absbigg}[1]{\biggl \lvert #1 \biggr\rvert}
\newcommand{\absBigg}[1]{\Biggl \lvert #1 \Biggr\rvert}

\newcommand{\hessian}{\mathrm{H}}

\makeatletter
\@ifundefined{thmenvcounter}{}
{%
	\newtheorem{envcounter}{EnvcounterDummy}[\thmenvcounter]
	
	\newtheorem{proposition}[envcounter]{Proposition}
	\newtheorem{lemma}[envcounter]{Lemma}
	
	\newtheorem{remark}[envcounter]{Remark}
	\newtheorem{example}[envcounter]{Example}

	\newtheorem{assumption}[envcounter]{Assumption}
}
\makeatother

\newsiamthm{example}{Example}
\newsiamthm{remark}{Remark}
\newsiamthm{assumption}{Assumption}

\headers{Stochastic filtering with moment representation}{Zheng Zhao and Juha Sarmavuori}

\title{Stochastic filtering with moment representation
	\thanks{Submitted to the editors DATE.
		\funding{This research was partially supported by the Wallenberg AI, Autonomous Systems and Software Program (WASP) funded by Knut and Alice Wallenberg Foundation. The computations handling was enabled by resources provided by the National Academic Infrastructure for Supercomputing in Sweden (NAISS) and the Swedish National Infrastructure for Computing (SNIC) partially funded by the Swedish Research Council through grant agreements no. 2022-06725 and no. 2018-05973.}}
}

\author{Zheng Zhao\thanks{Department of Information Technology, Uppsala University, Sweden (\email{zheng.zhao@it.uu.se}).}
	\and Juha Sarmavuori\thanks{Department of Electrical Engineering and Automation, Aalto University, Finland (\email{juha.sarmavuori@aalto.fi}).}
}

\ifpdf
\hypersetup{
	pdftitle={Stochastic filtering with moment representation},
	pdfauthor={Zheng Zhao and Juha Sarmavuori}
}
\fi

\begin{document}
	
	\maketitle
	
	\begin{abstract}
		Stochastic filtering refers to estimating the probability distribution of the latent stochastic process conditioned on the observed measurements in time. In this paper, we introduce a new class of convergent filters that represent the filtering distributions by their moments. The key enablement is a quadrature method that uses orthonormal polynomials spanned by the moments. We prove that this moment-based filter is asymptotically exact in the order of moments, and show that the filter is also computationally efficient and is in line with the state of the art.
	\end{abstract}
	
	
	
	\section{Introduction}
	\label{sec:intro}
	In this manuscript, we study the filtering problem concerned with models of the form
	\begin{equation}
		\begin{split}
			\diff X(t) &= a(X(t)) \diff t + b(X(t)) \diff W(t),\\
			X_0 &\sim \PP_{X_0}, \\
			Y_k \cond X_k &\sim p_{Y_k \cond X_k},
			\label{equ:model-formulation}
		\end{split}
	\end{equation}
	where the process $\lbrace X(t) \in \R^d \colon t\geq 0 \rbrace$ solves the It\^{o} stochastic differential equation (SDE) defined by a standard Wiener process $\lbrace W(t) \in \R^{d_w} \colon t\geq 0 \rbrace$, drift function $a \colon \R^d \to \R^d$, dispersion function $b \colon \R^d \to \R^{d \times d_w}$, and initial distribution $\PP_{X_0}$. The random variable $Y_k \in \R^{d_y}$ stands for the measurement of $X_k \coloneqq X(t_k)$ at any discrete time $t_k$ following a given conditional probability density function (PDF) $p_{Y_k \cond X_k}$. In addition, if $X$ is a discrete-time process instead of a solution to the SDE above, we only require that the conditional expectation $\expec{g(X_k) \cond X_{k-1}}$ is computable for any polynomial $g$.\looseness=-1 
	
	The filtering problem refers to solving the probability distribution $\PP_{X_k \cond Y_{1:k}}$ of $X_k$ conditioned on the collection of measurements $Y_{1:k} \coloneqq \lbrace Y_1, Y_2, \ldots, Y_k \rbrace$ for $k=1,2,\ldots$. This is a classical problem, and it is known to be challenging to compute the exact solution except for some isolated models. In the literature, there are plenty of approximate methods, such as Gaussian filters~\cite{Kazufumi2000} and particle filters~\cite{Chopin2020}, which are arguably the most popular ones. The principle of Gaussian filters (e.g., extended Kalman filters and Gauss--Hermite filters) is to approximate the filtering distribution $\PP_{X_k \cond Y_{1:k}}$ by a Gaussian so that the filtering problem boils down to only computing the mean and covariance of $\PP_{X_k \cond Y_{1:k}}$, which are usually efficient to compute. However, Gaussian approximations do not converge to the actual distribution that is non-Gaussian. In a different flavour, particle filters resort to approximate the filtering distribution by weighted samples, and then make use of sequential Monte Carlo techniques to estimate these weights and samples recursively in time. Under mild system conditions, the particle filters are convergent in the number of samples~\cite{Bain2009, Chopin2020}. However, to converge fast, we need a large number of samples, which in turn make the filtering routine computationally demanding and memory-consuming. It is also a common problem that sequential Monte Carlo methods can produce degenerate or impoverished weights, and solving such problems often incurs additional computations.
	
	There are also convergent filters by approximating the solution to the Kushner--Stratonovich equation. Examples are projection filters~\cite{Brigo1998, Emzir2023} which project the filtering densities onto finite-dimensional manifolds. However, these filters are primarily concerned with continuous-time measurements, while we focus on the discrete-time setting. Hence, we do not relate them to the scope of this paper. For detailed reviews of stochastic filters and their properties, we refer the readers to, for instance, \cite{Jazwinski1970, Bain2009, Sarkka2013, Law2015}.

	\subsection{Contributions}
	\label{sec:contributions}
	We introduce a new class of asymptotically exact and efficient filters to solve the filtering problem in Equation~\eqref{equ:model-formulation}. Specifically, we represent the filtering distribution $\PP_{X_k \cond Y_{1:k}}$ by a sequence of its moments, and then we recursively estimate this sequence for $k=1,2,\ldots$ by using a moment-based quadrature method. To expose the idea, let us suppose that $X$ is unidimensional (i.e., $d=1$). Then, at each time $t_k$, we use a sequence of $2 \, N$ moments
	\begin{equation}
		\begin{split}
			M_k^N &\coloneqq \lbrace m_{k, 0},  m_{k, 1},  m_{k, 2}, \ldots, m_{k, 2 \, N - 1} \rbrace, \\
			m_{k, n} &\coloneqq \expec{X_k^n \cond Y_{1:k}} \coloneqq \int x^n \diff \PP_{X_k \cond Y_{1:k}}(x),
			\label{equ:moments-definition-1d}
		\end{split}
	\end{equation}
	to approximately represent $\PP_{X_k \cond Y_{1:k}}$. Note that $m_{k, 0} = 1$ by definition. This moment-based representation converges in distribution as $N\to\infty$, if the target distribution $\PP_{X_k \cond Y_{1:k}}$ is determined by its moments (see, e.g., \cite[Chap. 15]{AchimKlenke2014} for sufficient conditions). Now suppose that we explicitly know the moments $M_0^N$ of the initial $\PP_{X_0 \cond Y_{1:0}} \coloneqq \PP_{X_0}$. We show a quadrature method so that we can use $M_0^N$ to approximate the moments $M_1^N$ of the next filtering distribution $\PP_{X_1 \cond Y_{1:1}}$, denoted by $\widehat{M}_1^N$. Likewise, we then continue to compute $\widehat{M}_k^N$ based on $\widehat{M}_{k-1}^N$, and so forth for $k = 2,3,\ldots$. This quadrature method generates the quadrature rules based on the orthonormal polynomials spanned by the moments, which is in a similar spirit as the Golub--Welsch approach~\cite{Golub1969}.
	
	Our proposed moment filter is a significant contribution to the community in terms of convergence and computation. We prove that at any time $t_k$, the approximate moments $\widehat{M}_k^N$ converge to the true moments as $N\to\infty$ under mild conditions of the system. Moreover, if the true filtering distributions are determined by their moments, then the discrete measures generated by the approximate moments and the quadrature method converge weakly to the truth too. In addition, the moment filter simultaneously gives a consistent and differentiable likelihood approximation which we can use to estimate the model parameters by maximum likelihood. The experiments show that the moment filter converges numerically as $N\to\infty$, and that the convergence speed is substantially faster than that of standard particle filters. While controlling the particle filter to have a similar computation time as the moment filter, the moment filter is significantly more accurate than the particle filter by a few orders of magnitude. 

	\subsection{Structure}
	\label{sec:structure}
	The paper is structured as follows. In Section~\ref{sec:quadrature}, we present the unidimensional quadrature method based on moments, and then we generalise the method for multidimensional systems. In Section~\ref{sec:filtering}, we derive the filter with moment representations, and then we show how to apply the introduced moment quadrature method to the moment filter. In the same section, we prove that the moment filter converges in distribution and moments. The numerical results are shown in Section~\ref{sec:experiments}, followed by conclusions and discussions in Section~\ref{sec:conclusions}. Finally, in Section~\ref{sec:related-works}, we discuss the related works for comparison to our method. 
	
	\section{Quadrature with moments}
	\label{sec:quadrature}
	Let $X\in\R^d$ be a random variable and $\PP$ be its probability distribution/measure. For clarity, we for now assume that the dimension $d=1$. The essence of the filtering problem that we aim to solve consists in computing the integral
	\begin{equation*}
		\expec{f(X)} \coloneqq \int f(x) \diff \PP(x) \approx \sum^{N}_{n=1} w_n \, f(\lambda_n),
	\end{equation*}
	by a set of quadrature rules $\lbrace w_n, \lambda_n \rbrace_{n=1}^{N}$, and the quadrature rules are to be determined by the moments $m_n \coloneqq \expec{X^n}$ of the measure $\PP$. A straightforward solution is to approximate the integrand $f$ by a power series. Then the integral is approximated by a sum of moments weighted by the derivatives of $f$. However, this approach has limited applications, as it	requires the integrand to be analytic which is a restrictive condition. Moreover, computing high-order derivatives is computationally demanding.\looseness=-1 
	
	To solve this moment-quadrature problem, we formulate a system of orthonormal polynomials, the coefficients of which are defined by the moments, and then we use the roots of the polynomial with the highest degree as the quadrature nodes~\cite{Golub1969, Golub2010}. By using this approach, the quadrature approximation is exact for any polynomial $f$ of degrees equal to or less than $2 \, N - 1$ with moments $M^N \coloneqq \lbrace m_n \rbrace_{n=0}^{2 \, N - 1}$. This in turn means that the quadrature is asymptotically exact in $N$ for any continuous $f$ on a compact domain by Weierstrass theorem. We detail this approach in the following. 
	
	Let us define an inner product $\innerp{f, g} \coloneqq \int f \, g \diff \PP = \expec{f(X) \, g(X)}$, and denote $\psi_n$ an orthonormal polynomial of degree $n$. It is well-known that any orthonormal polynomial system $\psi = \lbrace \psi_0, \psi_1, \ldots, \psi_N\rbrace$ in terms of this inner product is uniquely characterised by a three-term recurrence relation~\cite{Gautschi2004, Golub2010} 
	\begin{equation}
		\beta_{n+1} \, \psi_{n+1}(x) = (x - \alpha_{n+1}) \, \psi_n(x) - \beta_n \, \psi_{n-1}(x), \quad n=0,1,\ldots,N-1,
		\label{equ:three-term}
	\end{equation}
	where
	\begin{equation*}
		\begin{split}
			\psi_{-1}(x) &\coloneqq 0, \quad \psi_0(x) = 1, \quad \text{for all } x, \\
			\alpha_n &= \innerp{\psi_{n-1}, Z\psi_{n-1}}, \quad n=1,2,\ldots,N,\\
			\beta_0 = 1, \quad \beta_n &=\innerp{\psi_n, Z\psi_{n-1}},\quad n=1,2,\ldots, N-1, 
		\end{split}
	\end{equation*}
	and $Z$ is a multiplication operator defined via $(Zg)(x) \coloneqq x \, g(x)$ for any function $g$. If we rewrite the three-term recurrence relation in a vector form~\cite[pp. 86]{Golub2010}, the coefficients $\lbrace \alpha_n \rbrace_{n=1}^N$ and $\lbrace \beta_n \rbrace_{n=1}^{N-1}$ constitute a tridiagonal Jacobi matrix
	\begin{equation}
		J_N \coloneqq 
		\begin{bmatrix}
			\alpha_1 & \beta_1 &  &  & 0\\
			\beta_1 & \alpha_2 & \beta_2 &  &  \\
			& \beta_2 & \ddots & \ddots &  \\
			& & \ddots & \ddots & \beta_{N-1}\\
			0 & & & \beta_{N-1} & \alpha_N
		\end{bmatrix}.
		\label{equ:jacobi-matrix}
	\end{equation}
	A classical result by~\cite{Golub1969} shows that the roots $\lbrace \lambda_n \rbrace_{n=1}^N$ of $\psi_N$ are the eigenvalues of $J_N$, and that the corresponding quadrature weights $\lbrace w_n \rbrace_{n=1}^N$ are the squares of the first components of the eigenvectors of $J_N$ (i.e., if $u_n$ is the $n$-th eigenvector of $J_N$, then $w_n = u_{n,1}^2$, where $u_{n,1}$ is the first component of $u_n$). Computing the eigendecomposition of $J_N$ is notably efficient, since the Jacobi matrix is tridiagonal. 
	
	To find such an orthonormal system whose coefficients $\lbrace \alpha_n \rbrace_{n=1}^N$ and $\lbrace \beta_n \rbrace_{n=1}^{N-1}$ are determined by the moments, we define a system of linearly independent functions $\phi \coloneqq \lbrace \phi_0, \phi_1, \ldots, \phi_{N-1} \rbrace$, where $\phi_n(x) \coloneqq x^n$. Observing that
	\begin{equation*}
		\begin{split}
			G_N &\coloneqq \\
			&\begin{bmatrix}
				\innerp{\phi_0, \phi_0} & \innerp{\phi_0, \phi_1} & \cdots & \innerp{\phi_0, \phi_{N-1}} \\
				\innerp{\phi_1, \phi_0} & \innerp{\phi_1, \phi_1} & \cdots & \innerp{\phi_1, \phi_{N-1}} \\
				\vdots & \vdots & \ddots & \vdots \\
				\innerp{\phi_{N-1}, \phi_0} & \innerp{\phi_{N-1}, \phi_1} & \cdots & \innerp{\phi_{N-1}, \phi_{N-1}}
			\end{bmatrix}
			=
			\begin{bmatrix}
				m_0 & m_1 & \cdots & m_{N-1} \\
				m_1 & m_2 & \cdots & m_N \\
				\vdots & \vdots & \ddots & \vdots \\
				m_{N-1} & m_N & \cdots & m_{2 \, N - 2}
			\end{bmatrix}
		\end{split}
	\end{equation*}
	is a Gram/Hankel matrix of moments, we can then obtain the desired orthonormal system $\psi$ by a Gram--Schimidt orthonormalisation of $\phi$. The results in~\cite{Golub1969} show a straightforward routine to compute the coefficients in Equation~\eqref{equ:three-term} by the elements of the Cholesky decomposition of $G_N$. However, this approach loses two degrees of exactness, since $G_N$ does not contain the terminal moment $m_{2 N - 1}$. That is, with $2 \, N$ moments $M^N$, the method results in a Jacobi matrix of size $N-1$ which makes the quadrature approximation exact for polynomial integrands of degrees equal to or less than $2 \, N - 3$. To compensate the exactness to up to degree $2 \, N - 1$, we can see the Jacobi coefficients as a matrix representation of the multiplication operator $Z$ in $\psi$, at the cost of additional computations for solving a linear system~\cite{Sarmavuori2019}. To see this, we define another Hankel matrix of moments
	\begin{equation*}
		\begin{split}
			H_N &\coloneqq
			\begin{bmatrix}
				\innerp{\phi_0, Z\phi_0} & \innerp{\phi_0, Z\phi_1} & \cdots & \innerp{\phi_0, Z\phi_{N-1}} \\
				\innerp{\phi_1, Z\phi_0} & \innerp{\phi_1, Z\phi_1} & \cdots & \innerp{\phi_1, Z\phi_{N-1}} \\
				\vdots & \vdots & \ddots & \vdots \\
				\innerp{\phi_{N-1}, Z\phi_0} & \innerp{\phi_{N-1}, Z\phi_1} & \cdots & \innerp{\phi_{N-1}, Z\phi_{N-1}}
			\end{bmatrix}\\
			&=
			\begin{bmatrix}
				m_1 & m_2 & \cdots & m_N \\
				m_2 & m_3 & \cdots & m_{N+1} \\
				\vdots & \vdots & \ddots & \vdots \\
				m_N & m_{N+1} & \cdots & m_{2 \, N - 1}
			\end{bmatrix}
		\end{split}
	\end{equation*}
	which is the finite matrix representation of the multiplication operator $Z$ in $\phi$. Now let $L_N \, L_N^\trans = G_N$ be the Cholesky decomposition of $G_N$, we can then transform the matrix representation of the operator $Z$ in $\phi$ to $\psi$ by
	\begin{equation}
		\begin{split}
			L_N^{-1} \, H_N^{\phantom{\trans}} \, (L_N^\trans)^{-1}
			&=
			\begin{bmatrix}
				\innerp{\psi_0, Z\psi_0} & \innerp{\psi_0, Z\psi_1} & \cdots & \innerp{\psi_0, Z\psi_{N-1}} \\
				\innerp{\psi_1, Z\phi_0} & \innerp{\psi_1, Z\psi_1} & \cdots & \innerp{\psi_1, Z\psi_{N-1}} \\
				\vdots & \vdots & \ddots & \vdots \\
				\innerp{\psi_{N-1}, Z\psi_0} & \innerp{\psi_{N-1}, Z\psi_1} & \cdots & \innerp{\psi_{N-1}, Z\psi_{N-1}}
			\end{bmatrix}\\
			&= J_N
			\label{equ:jacobi-matrix-from-H}
		\end{split}
	\end{equation}
	which equals to the Jacobi matrix in Equation~\eqref{equ:jacobi-matrix} by definition. 
	
	In summary, with moments $M^N$, we first use the moments to form the matrices $G_N$ and $H_N$, and then we take the Cholesky decomposition of $G_N$ and solve the linear system as per Equation~\eqref{equ:jacobi-matrix-from-H} to compute $J_N$. With the Jacobi matrix $J_N$, we compute its eigenvalues and eigenvectors that we use to determine the quadrature rules. The computational complexity of this quadrature method is dominated by the Cholesky decomposition which is of $O(N^3)$. In Figure~\ref{fig:1d-quadrature-rules}, we exemplify the quadrature rules generated for three distributions with $N=11$.  
	
	\begin{figure}
		\centering
		\includegraphics[width=\linewidth]{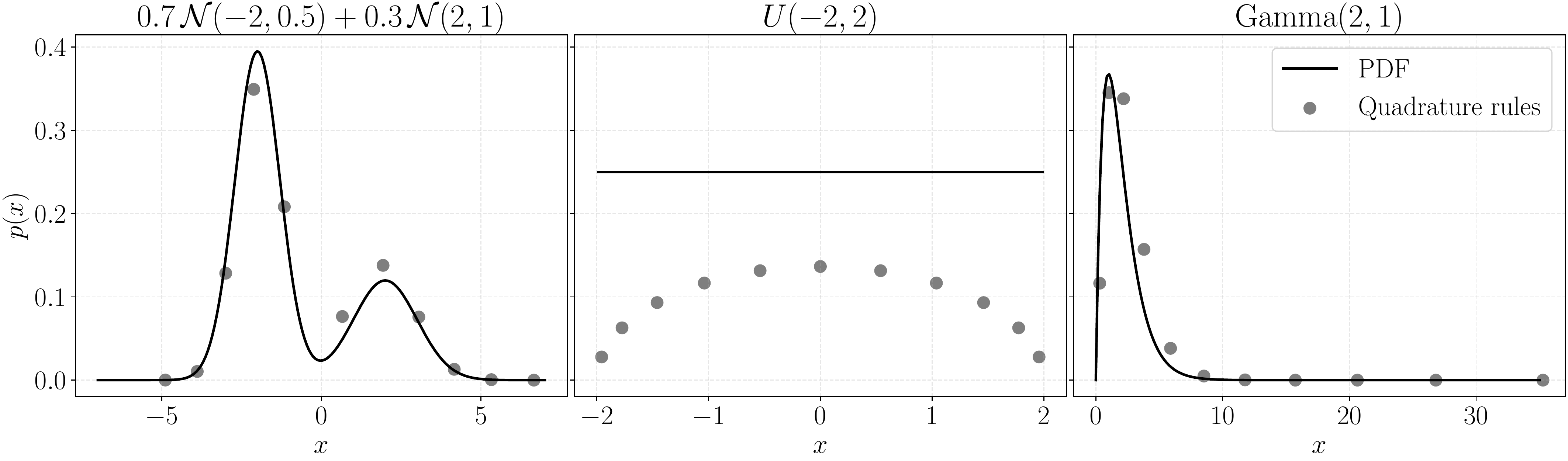}
		\caption{Quadrature rules ($N = 11$) for a Gaussian sum (left), a uniform (middle), and a Gamma (right) distribution. The horizontal and vertical locations of the grey points represent the values of the quadrature nodes and weights, respectively. We see that the quadrature rules essentially form a discrete approximation to the continuous distribution, where the discrete probability bins are chosen such that the expectation is exact for up to $2 \, N - 1$ degree polynomials. }
		\label{fig:1d-quadrature-rules}
	\end{figure}
	
	\subsection{Generalisation for multidimensional quadrature}
	\label{sec:nd-quadrature}
	Now consider that the dimension $d > 1$. The quadrature method in the previous section no longer applies, since the coefficients of the three-term recurrence relation become vector-valued~\cite[Chap. 3]{Dunkl2014}. To generalise the quadrature for multidimensional integrations, we resort to view the numerical integration as a finite matrix approximation to a multiplication operator. More specifically, for each dimension, we define a multiplication operator, and then we analogously compute the Jacobi matrix associated for that operator. The resulting quadrature rules are finally given by the Cartesian products of the eigendecompositions of these Jacobi matrices. To detail this generalisation, we introduce the following technical prerequisites.
	
	Let $X_{(i)}$ denote the $i$-th element of the vector $X\in\R^d$. To define the moments for multidimensional random variables, we introduce multi-index $\cu{n} = (n_1, n_2, \ldots, n_d)$ of fixed length $d$, and the exponent $X^\cu{n}$ reads as the product $X^\cu{n} \coloneqq X_{(1)}^{n_1} \, X_{(2)}^{n_2} \cdots X_{(d)}^{n_d}$. As an example, if $d=3$ and $\cu{n} = (5, 1, 4)$, then $X^\cu{n} = X_{(1)}^5 \, X_{(2)}^{\phantom{1}} \, X_{(3)}^4$. The collection of moments that we need to generate the quadrature rules is then
	\begin{equation*}
		\begin{split}
			M^N \coloneqq \lbrace m_\cu{n} \colon \abs{\cu{n}} \leq 2 \, N - 1 \rbrace, \quad m_\cu{n} \coloneqq \expec{X^\cu{n}},
		\end{split}
	\end{equation*}
	where $\abs{\cu{n}}$ stands for the sum of the multi-index $\cu{n}$. The collection $M^N$ has in total $\binom{2 \, N - 1 + d}{2 \, N - 1}$ elements, and the elements can be arbitrarily ordered.
	
	Similarly as in Section~\ref{sec:quadrature}, we define a system of linearly independent basis $\phi \coloneqq \lbrace \phi_\cu{n} \colon \abs{\cu{n}} \leq N - 1 \rbrace$, where $\phi_\cu{n}(x) \coloneqq x^\cu{n}$, so that they represent the moments. Furthermore, to ensure that our quadrature approximation is valid, we restrict the first basis function in $\phi$ to be $\phi_{\cu{n}_0}(x)=1$, where the multi-index $\cu{n}_0$ has sum $\abs{\cu{n}_0}=0$. The Gram matrix induced by this system is denoted by $G_S$ with element $(G_{S})_{ij} \coloneqq \innerp{\phi_{\cu{n}_i}, \phi_{\cu{n}_j}}$ for $i,j=0, 1,\ldots, S-1$, where the size of the matrix is $S=\binom{N - 1 + d}{N - 1}$. Since the positive definiteness of the Gram matrix is independent of how we order the basis functions in $\phi$, a convenient choice is the graded lexicographical order~\cite{Dunkl2014} which has a few nice properties in favour of indexing and computation.
	
	For each dimension $i=1,2,\ldots, d$, we introduce a multiplication operator $Z_i$ defined via $(Z_i g)(x) \coloneqq x_{(i)} \, g(x)$ for any function $g\colon \R^d \to\R$. The finite matrix representation of the operator $Z_i$ in $\phi$ is defined by
	\begin{equation*}
		H_{S, i} 
		\coloneqq
		\begin{bmatrix}
			\innerp{\phi_{\cu{n}_0}, Z_i\phi_{\cu{n}_0}} & \innerp{\phi_{\cu{n}_0}, Z_i\phi_{\cu{n}_1}} & \cdots & \innerp{\phi_{\cu{n}_0}, Z_i\phi_{\cu{n}_{S-1}}} \\
			\innerp{\phi_{\cu{n}_1}, Z_i\phi_{\cu{n}_0}} & \innerp{\phi_{\cu{n}_1}, Z_i\phi_{\cu{n}_1}} & \cdots & \innerp{\phi_{\cu{n}_1}, Z_i\phi_{\cu{n}_{S-1}}} \\
			\vdots & \vdots & \ddots & \vdots \\
			\innerp{\phi_{\cu{n}_{S-1}}, Z_i\phi_{\cu{n}_0}} & \innerp{\phi_{\cu{n}_{S-1}}, Z_i\phi_{\cu{n}_1}} & \cdots & \innerp{\phi_{\cu{n}_{S-1}}, Z_i\phi_{\cu{n}_{S-1}}}
		\end{bmatrix}.
	\end{equation*}
	
	\begin{example}
		\label{example:nd-gram-hankel}
		Consider $d=2$ and $N=2$. Let us choose graded lexicographical ordered multi-indices $\cu{n}_0 = (0, 0), \cu{n}_1 = (0, 1), \cu{n}_2 = (1, 0), \ldots, \cu{n}_9 = (3, 0)$, so $\phi_{\cu{n}_0}(x)=1$, $\phi_{\cu{n}_1}(x)=x_{(2)}$, and $\phi_{\cu{n}_2}(x)=x_{(1)}$. Suppose that the ten moments in $M^N$ are ordered by these multi-indices as well. The Gram matrix and the matrices of the two multiplication operators are then
		\begin{equation*}
			\begin{split}
				G_3 &= 
				\begin{bmatrix}
					1 & \expec{X_{(2)}} & \expec{X_{(1)}} \\
					\expec{X_{(2)}} & \expec{X_{(2)}^2} & \expec{X_{(1)} \, X_{(2)}}\\
					\expec{X_{(1)}} & \expec{X_{(1)} \, X_{(2)}} & \expec{X_{(1)}^2}
				\end{bmatrix}
				= 
				\begin{bmatrix}
					m_{\cu{n}_0} & m_{\cu{n}_1} & m_{\cu{n}_2} \\
					m_{\cu{n}_1} & m_{\cu{n}_3} & m_{\cu{n}_4} \\
					m_{\cu{n}_2} & m_{\cu{n}_4} & m_{\cu{n}_5}
				\end{bmatrix},\\
				H_{3, 1} &= 
				\begin{bmatrix}
					\expec{X_{(1)}} & \expec{X_{(1)} \, X_{(2)}} & \expec{X_{(1)}^2} \\
					\expec{X_{(1)} \, X_{(2)}} & \expec{X_{(1)} \, X_{(2)}^2} & \expec{X_{(1)}^2 \, X_{(2)}} \\
					\expec{X_{(1)}^2} & \expec{X_{(1)}^2 \, X_{(2)}} & \expec{X_{(1)}^3}
				\end{bmatrix}
				=
				\begin{bmatrix}
					m_{\cu{n}_2} & m_{\cu{n}_4} & m_{\cu{n}_5} \\
					m_{\cu{n}_4} & m_{\cu{n}_7} & m_{\cu{n}_8} \\
					m_{\cu{n}_5} & m_{\cu{n}_8} & m_{\cu{n}_9}
				\end{bmatrix},\\
				H_{3, 2} &= 
				\begin{bmatrix}
					\expec{X_{(2)}} & \expec{X_{(2)}^2} & \expec{X_{(1)} \, X_{(2)}} \\
					\expec{X_{(2)}^2} & \expec{X_{(2)}^3} & \expec{X_{(1)} \, X_{(2)}^2} \\
					\expec{X_{(1)} \, X_{(2)}} & \expec{X_{(1)} \, X_{(2)}^2} & \expec{X_{(1)}^2 \, X_{(2)}}
				\end{bmatrix}
				=
				\begin{bmatrix}
					m_{\cu{n}_1} & m_{\cu{n}_3} & m_{\cu{n}_4} \\
					m_{\cu{n}_3} & m_{\cu{n}_6} & m_{\cu{n}_7} \\
					m_{\cu{n}_4} & m_{\cu{n}_7} & m_{\cu{n}_8}
				\end{bmatrix}.
			\end{split}
		\end{equation*}
	\end{example}
	
	Let $L_S \, L_S^\trans = G_S$ be the Cholesky decomposition of the Gram matrix $G_S$, then we can compute the matrix representation $\vscriptalign{\mathring{H}}_{S, i}$ of $Z_i$ in an orthonormal basis system $\psi \coloneqq \lbrace \psi_{\cu{n}_0}, \psi_{\cu{n}_1}, \ldots, \psi_{\cu{n}_{S-1}} \rbrace$ by
	\begin{equation}
		\vscriptalign{\mathring{H}}_{S, i} = L_S^{-1} \, H_{S, i}^{\phantom{\trans}} \, \bigl(L_S^\trans\bigr)^{-1},
		\label{equ:nd-hankel-orth}
	\end{equation}
	where its $u,v$-th matrix element $(\vscriptalign{\mathring{H}}_{S, i})_{uv} = \innerp{\psi_{\cu{n}_u}, Z_i\psi_{\cu{n}_v}}$, and $\psi_{\cu{n}_0}(x) = 1$.
	
	Let $f\colon \R^d \to\R$ be any continuous function on a compact domain, and $A$ be any self-adjoint operator. We define the functional operator $f(A)$ as a spectral integral $f(A) \coloneqq \int_{\sigma(A)} f(z) \diff P_A(z)$, where $P_A$ is a projection-valued measure induced by $A$, and $\sigma(A)$ is the spectrum of $A$. Hence, by $f(Z_1, Z_2, \ldots, Z_d)$ we mean that it is an operator $f(Z_1, Z_2, \ldots, Z_d) = \int f(z_1, z_2, \ldots, z_d) \diff P_{Z_1}(z_1)\diff P_{Z_2}(z_2)\cdots \diff P_{Z_d}(z_d)$. For details of these definitions, see, for instance, \cite{Simon2015}. It turns out that $f(Z_1, Z_2, \ldots, Z_d)$ is a multiplication operator as well, that is, for any function $g \colon \R^d \to \R$, 
	\begin{equation*}
		\bigl( f(Z_1, Z_2, \ldots, Z_d)g \bigr)(x) = f(x) \, g(x).
	\end{equation*}
	With the property above in mind, we can now think of the quadrature as a finite-matrix approximation to the operator $f(Z_1, Z_2, \ldots, Z_d)$. Specifically, the catch is to represent the integral as
	\begin{equation}
		\begin{split}
			\int f(x) \diff \PP(x) = \innerp{\psi_{\cu{n}_0}, f \, \psi_{\cu{n}_0}} &= \innerp{\psi_{\cu{n}_0}, f(Z_1, Z_2, \ldots, Z_d) \, \psi_{\cu{n}_0}} \\
			&\approx e_0^\trans \, f(\vscriptalign{\mathring{H}}_{S, 1}, \vscriptalign{\mathring{H}}_{S, 2}, \ldots, \vscriptalign{\mathring{H}}_{S, d}) \, e_0,
			\label{equ:integral-as-matrix}
		\end{split}
	\end{equation}
	where $f(\vscriptalign{\mathring{H}}_{S, 1}, \vscriptalign{\mathring{H}}_{S, 2}, \ldots, \vscriptalign{\mathring{H}}_{S, d}) \in \R^{S \times S}$ is a matrix that approximately represents the operator $f(Z_1, Z_2, \ldots, Z_d)$, and $e_0 = \begin{bmatrix} 1 & 0 & \cdots & 0 \end{bmatrix}^\trans$ extracts the first component of the matrix. Since we have defined $f(Z_1, Z_2, \ldots, Z_d)$ as a spectral integral, the definition of the matrix $f(\vscriptalign{\mathring{H}}_{S, 1}, \vscriptalign{\mathring{H}}_{S, 2}, \ldots, \vscriptalign{\mathring{H}}_{S, d})$ is
	\begin{equation}
		\begin{split}
			&f(\vscriptalign{\mathring{H}}_{S, 1}, \vscriptalign{\mathring{H}}_{S, 2}, \ldots, \vscriptalign{\mathring{H}}_{S, d}) \\
			&\coloneqq \sum_{n_1=1}^S \sum_{n_2=1}^S \cdots \sum_{n_d=1}^S f(\lambda_{n_1}, \lambda_{n_2}, \ldots, \lambda_{n_d}) \, u_{n_1}^{\phantom{\trans}} \, u_{n_1}^\trans \, u_{n_2}^{\phantom{\trans}} \, u_{n_2}^\trans \, \cdots \, u_{n_d}^{\phantom{\trans}} \, u_{n_d}^\trans,
			\label{equ:f-Ho}
		\end{split}
	\end{equation}
	where $\lambda_{n_i}$ and $u_{n_i}$ are the $n_i$-th eigenvalue and eigenvector of the matrix $\vscriptalign{\mathring{H}}_{S, i}$, respectively. Now by substituting Equation~\eqref{equ:f-Ho} back into Equation~\eqref{equ:integral-as-matrix}, we see that the quadrature nodes are all the combinations of the eigenvalues, and the corresponding weights are products of inner products. More precisely, the quadrature is
	\begin{equation}
		\begin{split}
			&\int f(x) \diff \PP(x) \\
			&\approx \sum_{n_1=1}^S \sum_{n_2=1}^S \cdots \sum_{n_d=1}^S f(\lambda_{n_1}, \lambda_{n_2}, \ldots, \lambda_{n_d}) \, e_0^\trans \, u_{n_1}^{\phantom{\trans}} \, u_{n_1}^\trans \, u_{n_2}^{\phantom{\trans}} \, u_{n_2}^\trans \, \cdots \, u_{n_d}^{\phantom{\trans}} \, u_{n_d}^\trans \, e_0\\
			&\coloneqq \sum_{\cu{n} \in \mathfrak{n}_{N, d}} w_{\cu{n}} \, f(\lambda_{\cu{n}}), 
			\label{equ:nd-quadratrure-rules-compact}
		\end{split}
	\end{equation}
	where in the last line we compactly write the quadrature rules as
	\begin{equation}
		\begin{split}
			\lambda_\cu{n} &\coloneqq 
			\begin{bmatrix}
				\lambda_{n_1} & \lambda_{n_2} & \cdots & \lambda_{n_d}
			\end{bmatrix}^\trans, \\
			w_\cu{n} &\coloneqq \innerp{e_0, u_{n_1}}_S \, \biggl( \prod_{i=1}^{d-1} \innerp{u_{n_i}, u_{n_{i+1}}}_S \biggl)	\innerp{u_{n_d}, e_0}_S, \quad \innerp{x, y}_S \coloneqq x^\trans \, y, \\
			\mathfrak{n}_{N, d} &\coloneqq \bigl\lbrace (1, 2, \ldots, S) \times \overset{d}{\cdots} \times (1, 2, \ldots, S) \bigr\rbrace.
			\label{equ:nd-quadratrure-rules}
		\end{split}
	\end{equation}
	If we let $d=1$, it is clear that this generalised quadrature reduces to the unidimensional quadrature in Section~\ref{sec:quadrature}. Furthermore, we show that this quadrature method is also exact for multivariate polynomials of degree equal to or less than $2 \, N - 1$. This is given in the following lemma.

	\begin{lemma}
		\label{lemma:poly-exactness}
		The quadrature in Equation~\eqref{equ:nd-quadratrure-rules-compact} is exact for every polynomial $f$ of degree equal to or less than $2 \, N - 1$.
	\end{lemma}
	\begin{proof}
	To show the exactness for polynomials of degree equal to or less than $2 \, N - 1$, it is enough to prove $\innerp{\psi_{\cu{n}_0}, Z^\cu{n}\psi_{\cu{n}_0}} = e_0^\trans \, \mathring{H}_S^\cu{n} \, e_0$ for all $\abs{\cu{n}}\leq 2 \, N - 1$, where $Z^\cu{n} \coloneqq \prod_{i=1}^d Z_i^{n_i}$ and $\mathring{H}_S^\cu{n} \coloneqq \prod_{i=1}^d \mathring{H}_{S, i}^{n_i}$. Evidently, this holds for $\abs{\cu{n}} = 0$. To prove this for $0< \abs{\cu{n}}\leq 2 \, N - 1$, we decompose $\cu{n} = \cu{u} + \cu{v}$ in the way that $0\leq\abs{\cu{u}}\leq N-1$ and $1\leq\abs{\cu{v}}\leq N$. By Parseval's identity we have
	\begin{equation}
		\begin{split}
			\innerp{\psi_{\cu{n}_0}, Z^\cu{n}\psi_{\cu{n}_0}} &= \sum_{\abs{\cu{q}}\geq 0} \innerp{\psi_{\cu{q}}, Z^\cu{u} \psi_{\cu{n}_0}} \, \innerp{\psi_{\cu{q}}, Z^\cu{v} \psi_{\cu{n}_0}}\\
			&= \sum_{\abs{\cu{q}}\leq N-1} \innerp{\psi_{\cu{q}}, Z^\cu{u} \psi_{\cu{n}_0}} \, \innerp{\psi_{\cu{q}}, Z^\cu{v} \psi_{\cu{n}_0}},
		\end{split}
		\label{equ:truncated-parseval}
	\end{equation}
	where we truncate the sum at $N-1$ because of the orthonormality (i.e., $\innerp{\psi_{\cu{q}},Z^{\cu{u}}\psi_{\cu{n}_0}}=0$ for $\abs{\cu{q}}>N-1$). Next we look at $\innerp{\psi_{\cu{q}},Z^{\cu{v}}\psi_{\cu{n}_0}}$ in Equation~\eqref{equ:truncated-parseval} for $\abs{\cu{q}}\leq N -1$ and $1\leq\abs{\cu{v}}\leq N$. We refer to the index of the first non-zero element of $\cu{v}$ as $i$, that is, $v_j=0$ for $j<i$, and $v_j\geq 0$ otherwise. Because $\abs{\cu{v}-\cu{e}_i}\leq N -1$, we can express the monomial $x\mapsto x^{\cu{v}-\cu{e}_i}$ as a linear combination of orthonormal polynomials up to order $N-1$, that is, $x^{\cu{v}-\cu{e}_i}=\sum_{\abs{\cu{p}}\leq N-1}c_{\cu{p}}\,\psi_{\cu{p}}(x)$ for some coefficients $c$. For finite matrix equivalent of the monomial of multiplication operators this means that
	\begin{equation}
		\begin{split}
			\innerp{\psi_{\cu{q}},Z^{\cu{v}}\psi_{\cu{n}_0}}
			&= \innerp{Z_i\,\psi_{\cu{q}},Z^{\cu{v}-\cu{e}_i}\psi_{\cu{n}_0}} \\
			&= \sum_{\abs{\cu{p}}\leq N - 1}c_{\cu{p}}\, \innerp{Z_i\psi_{\cu{q}}, \psi_{\cu{p}}}
			= \sum_{\abs{\cu{p}}\leq N-1} c_{\cu{p}}\,
			e_{j_{\cu{q}}}^\trans \, \mathring{H}_{S,i} \, e_{j_{\cu{p}}} \\
			&= \sum_{\abs{\cu{p}}\leq N-1} c_{\cu{p}}\,
			e_{j_{\cu{q}}}^\trans \,\mathring{H}_{S,i}\,\psi_{\cu{p}}(\mathring{H}_S) \, e_0 = e_{j_{\cu{q}}}^\trans\,\mathring{H}_{S,i}\,\mathring{H}_S^{\cu{v}-e_i}\,e_0
			=e_{j_{\cu{q}}}^\trans\,\mathring{H}_S^{\cu{v}}\,e_0,
		\end{split}
		\label{equ:one-more-order-matrix-elements}
	\end{equation}
	where $j_{\cu{q}}$ and $j_{\cu{p}}$ depend on the ordering of the orthonormal basis functions so that $\psi_{\cu{q}}$ and $\psi_{\cu{p}}$ are the $j_{\cu{q}}$-th and $j_{\cu{p}}$-th basis function, respectively, both ranging from 0 to $S-1$. 
	
	In order to preserve the ordering of the matrix computations in Equation~\eqref{equ:f-Ho}, we further refine the decomposition to $\cu{u}$ and $\cu{v}$ so that there is index $l$ such that $u_i=0$ for $i>l$ and $\cu{v}_j=0$ for $j<l$. Now by Equation~\eqref{equ:one-more-order-matrix-elements}, we can write Equation~\eqref{equ:truncated-parseval} in terms of the finite matrices
        as
        \begin{equation*}
          \innerp{\psi_{\cu{n}_0}, Z^\cu{n}\psi_{\cu{n}_0}}
          =
          \sum_{\abs{\cu{q}}\leq N-1} \innerp{\psi_{\cu{q}}, Z^\cu{u} \psi_{\cu{n}_0}} \, \innerp{\psi_{\cu{q}}, Z^\cu{v} \psi_{\cu{n}_0}}
          = 
          \sum_{i=0}^{S-1} e_0^\trans\,\mathring{H}_S^{\cu{u}}\,e_i\,
          e_i^\trans\,\mathring{H}_S^{\cu{v}}\,e_0 =
          e_0^\trans\,\mathring{H}_S^{\cu{n}}\,e_0,
        \end{equation*}
    for all $0<\abs{\cu{n}}\leq 2\,N-1$.
    \end{proof}

	Based on the exactness for polynomials, we can then conclude the convergence of the quadrature method for continuous functions in the following proposition.
	
	\begin{proposition}
		\label{prop:quarature-convergence}
		Let the compact support of $\PP$ be a subset of a compact hypercube $D \subset \R^d$, and let $f\colon D\to\R$ be any continuous function. Recall that $S=\binom{N - 1 + d}{N - 1}$. Then, 
		\begin{equation}
			\lim_{N\to\infty}\absbigg{\int_D f \diff \PP - \sum_{\cu{n} \in \mathfrak{n}_{N, d}} w_{\cu{n}} \, f(\lambda_{\cu{n}})} = 0.
			\label{equ:nd-quadrature-convergence}
		\end{equation}
	\end{proposition}
	\begin{proof}
		Since the measure $\PP$ has support on $D$, the quadrature nodes $\lbrace\lambda_\cu{n}\rbrace_{\cu{n}\in\mathfrak{n}_{N, d}}$ lie within $D$~\cite[Thm.~1]{Sarmavuori2019}. By Stone--Weierstrass theorem, for any $\epsilon>0$ there exists a polynomial $\rho_\epsilon$ on $D$ such that $\sup_{x\in D} \abs{f(x) - \rho_\epsilon(x)} < \epsilon$. Denote the residual $\tilde{f}_\epsilon \coloneqq f - \rho_\epsilon$, and recall that $\lbrace u_{n_i}^{\phantom{\trans}} \, u_{n_i}^\trans \rbrace_{i=1}^d$ are orthonormal projections. It follows that the Euclidean norm
		\begin{equation*}
			\begin{split}
				&\normBigg{\sum_{n_1=1}^S \sum_{n_2=1}^S \cdots \sum_{n_d=1}^S \tilde{f}(\lambda_{n_1}, \lambda_{n_2}, \ldots, \lambda_{n_d}) \, u_{n_1}^{\phantom{\trans}} \, u_{n_1}^\trans \, u_{n_2}^{\phantom{\trans}} \, u_{n_2}^\trans \, \cdots \, u_{n_d}^{\phantom{\trans}} \, u_{n_d}^\trans \, e_0}_2 \\
				&\leq \sup_{x \in D} \absbig{\tilde{f}(x)} < \epsilon.
			\end{split}
		\end{equation*}
		By Cauchy--Schwarz, the numerical quadrature of the residual $\tilde{f}$ is bounded:
		\begin{equation*}
			\begin{split}
				\absbig{I_N(\tilde{f})} &\coloneqq \absBigg{ e_0^\trans \sum_{n_1=1}^S \sum_{n_2=1}^S \cdots \sum_{n_d=1}^S \tilde{f}(\lambda_{n_1}, \lambda_{n_2}, \ldots, \lambda_{n_d}) \, u_{n_1}^{\phantom{\trans}} \, u_{n_1}^\trans \, u_{n_2}^{\phantom{\trans}} \, u_{n_2}^\trans \, \cdots \, u_{n_d}^{\phantom{\trans}} \, u_{n_d}^\trans \, e_0} \\
				&\leq \epsilon.
			\end{split}
		\end{equation*}
		The exact integration of $\tilde{f}$ is also bounded $I(\tilde{f}) \coloneqq \int_D \tilde{f} \diff \PP \leq \epsilon \expec{\cu{1}_D}$. Recall that for any polynomial $\rho_\epsilon$, we can always find an $N_\epsilon$ such that for every $N > N_\epsilon$, the numerical quadrature $I_N(\rho_\epsilon) = I(\rho_\epsilon)$ is exact. Therefore, we have
		\begin{equation*}
			\begin{split}
				\abs{I(f) - I_N(f)} &= \absbig{I(\tilde{f}) + I(\rho_\epsilon) - I_N(\tilde{f}) - I_N(\rho_\epsilon)} \\
				&\leq \absbig{I(\tilde{f})} + \absbig{I_N(\tilde{f})} + \abs{I(\rho_\epsilon) - I_N(\rho_\epsilon)} \leq (1 + \expec{\cu{1}_D}) \, \epsilon.\\
			\end{split}
		\end{equation*}
		This concludes the limit in Equation~\eqref{equ:nd-quadrature-convergence}.
	\end{proof}
	
	As a summary, with given moments $M^N = \lbrace m_\cu{n} \colon \abs{\cu{n}} \leq 2 \, N - 1 \rbrace$, we compute the quadrature rules as follows. We first select a partial order of the moments and the basis, and then we rearrange these moments into the Gram matrix $G_S$ and the Hankel matrices $H_{S, 1}, H_{S, 2}, \ldots, H_{S, d}$ (see, e.g., Example~\ref{example:nd-gram-hankel}). Then, we orthonormalise these Hankel matrices as per Equation~\eqref{equ:nd-hankel-orth} to obtain $\vscriptalign{\mathring{H}}_{S, 1}, \vscriptalign{\mathring{H}}_{S, 2}, \ldots, \vscriptalign{\mathring{H}}_{S, d}$. Finally, we compute the eigendecompositions of the orthonormalised Hankel matrices, and then combine the eigenvalues and eigenvectors as in Equation~\eqref{equ:nd-quadratrure-rules} for which we produce the quadrature rules. A pseudo-code of this quadrature method is given in the following algorithm. \looseness=-1

	\begin{algorithm2e}[h]
		\SetAlgoLined
		\DontPrintSemicolon
		\Func{\FnMQuad{$M^N$}}{%
			Build the Gram matrix $G_S$ and Hankel matrices $H_{S_1}, \ldots, H_{S_d}$ based on the moments in $M^N$ \;
			Cholesky decomposition $L_S \, L_S^\trans = G_S$\;
			\For(\tcp*[f]{In parallel}){$i=1$ \KwTo $d$}{%
				$\vscriptalign{\mathring{H}}_{S, i} = L_S^{-1} \, H_{S, i}^{\phantom{\trans}} \, (L_S^\trans)^{-1}$\;
				Compute eigenvalues and eigenvectors $\lbrace \lambda_{n_i}, u_{n_i} \rbrace_{n_i=1}^S$ of $\vscriptalign{\mathring{H}}_{S, i}$\;
			}
			\For(\tcp*[f]{In parallel}){$\cu{n}$ \KwForIn $\mathfrak{n}_{N, d}$}{%
				$\lambda_\cu{n} \coloneqq 
				\begin{bmatrix}
					\lambda_{n_1} & \lambda_{n_2} & \cdots & \lambda_{n_d}
				\end{bmatrix}^\trans$ \;
				$w_\cu{n} \coloneqq \innerp{e_0, u_{n_1}}_S \, \Bigl( \prod_{i=1}^{d-1} \innerp{u_{n_i}, u_{n_{i+1}}}_S \Bigl)	\innerp{u_{n_d}, e_0}_S$ \;
			}
			\KwRet{$\lbrace w_\cu{n}, \lambda_\cu{n} \rbrace_{n \in \mathfrak{n}_{S, d}}$}
		}
		\caption{$d$-dimensional moment quadrature with order $N$}
		\label{alg:nd-moment-quadrature}
	\end{algorithm2e}

	\begin{figure}[t!]
		\centering
		\includegraphics[width=\linewidth]{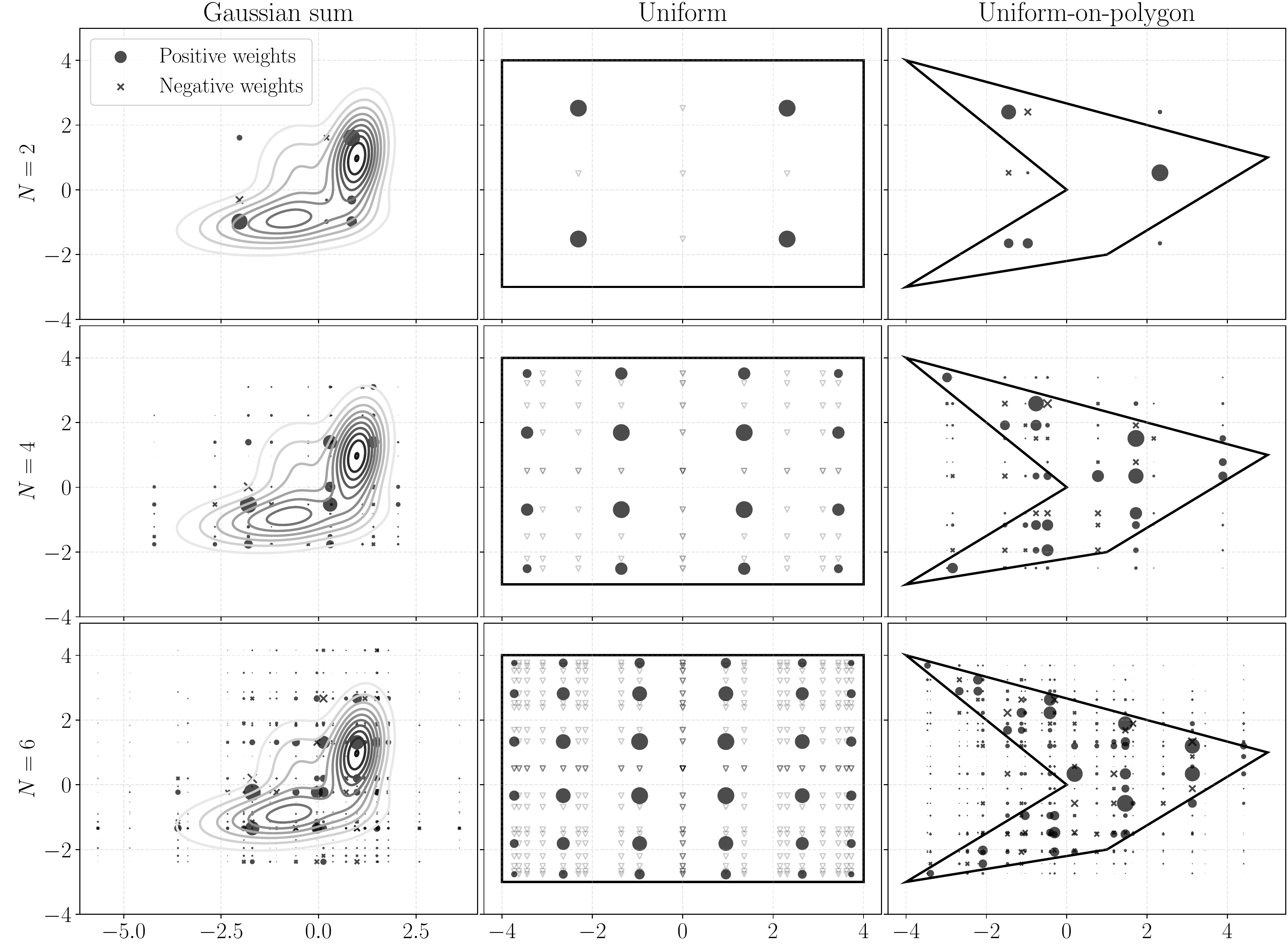}
		\caption{Quadrature rules for a Gaussian sum (left), a uniform (middle), and a uniform distribution on a polygon area (right). The location and size of the grey points represent the quadrature nodes and weights, respectively. For each plot, the sizes of the points are normalised by their maximum weight. In the middle, we mark the zero weights for the uniform distribution by triangles.}
		\label{fig:nd-quadrature-rules}
	\end{figure}

	Figure~\ref{fig:nd-quadrature-rules} exemplifies three two-dimensional distributions and their corresponding quadrature rules with different order $N$. We see that the quadrature nodes are always confined within a rectangle area due to the Cartesian product construction, and that the weights are noticeably larger in the high-density areas of the distributions than the low-density areas. This suggests that the quadrature rules generated from the moments can represent these distributions well to a reasonable extent. However, we also observe from the figure that there are negative weights, since Equation~\eqref{equ:nd-quadratrure-rules} does not guarantee non-negativity of the product of the inner products. With negative weights, the discrete measures generated by the moments are signed. The negative weights may also result in numerical instabilities (e.g., the quadrature for positive integrands may become negative), but on the other hand, it follows from Proposition~\ref{prop:quarature-convergence} that the sum of the negative weights admits an upper bound~\cite[Chap.~12.3, Thm.~8]{Krylov2005}.
	
	The uniform distribution in the middle of Figure~\ref{fig:nd-quadrature-rules} is an example that we can factorise the two-dimensional distribution into that of two independent random variables. We see that the majority of the nodes have zero weights, and that the nodes with non-zero weights agree with that of the product rule for independent variables~\cite[Chap.~5.6]{Davis1984}. The product rule is commonly used for constructing multidimensional numerical integration rules, for instance, the multidimensional Gauss--Hermite quadrature used in Gaussian filters~\cite{Kazufumi2000} and Edgeworth filters~\cite{Challa2000,Singer2008}. Now our new method generalises the product rule to cases where it is impossible to factorise the joint probability distribution into independent ones (e.g., the left and right sides of Figure~\ref{fig:nd-quadrature-rules}).
	
	The quadrature method in its current form is computationally demanding for high-dimensional integrations. Recall that $S = \binom{N - 1 + d}{N - 1}$, hence, if we fix $d$ then $S \sim O(N^d \, / \, d!)$. Due to the Cartesian product construction, the number of quadrature nodes is of $O\bigl(N^{d^2} \, / \, (d!)^d\bigr)$ which grows polynomially in $N$ of degree $d^2$. If we instead fix $N$, then $S \sim O(d^{N - 1} \, / \, (N - 1)!)$, thus, the number of quadrature nodes grows faster than the exponential speed in $d$. There are a few ways to reduce the number of quadrature nodes, for instance, by using L\'{a}nczos iterations to solve the matrix-vector multiplications in Equation~\eqref{equ:nd-quadratrure-rules-compact}. It also makes sense to come up with a sparse version, since Figure~\ref{fig:nd-quadrature-rules} shows that there are plenty of insignificant weights, in particular when $N$ is large. As an extreme example, if the probability distribution is elliptical and is thin along a direction, then the quadrature would be significantly inefficient, since the quadrature nodes spread in a rectangle. 
	
	\begin{remark}
		\label{remark:modified-moments}
		Using raw moments can lead to ill-conditioned Gram matrices~\cite{Golub2010}. To improve the condition number, we can make use of central or scaled moments. Specifically, if $(w, \lambda)$ is a pair of quadrature weight and node generated by the raw moments $M^N$, then $(w, \sigma \, \lambda + \mu)$ is the corresponding pair of weight and node generated by the scaled central moment $\bigl\lbrace \expec{((X - \mu) \, / \, \sigma)^\cu{n}} \colon \abs{\cu{n}} \leq 2 \, N - 1 \bigr\rbrace$, where $\mu$ and $\sigma\neq0$ can be the mean and any scaling factor, respectively.
	\end{remark}
	
	\section{Filtering with moment representation}
	\label{sec:filtering}
	Recall the definition of the filtering problem in Section~\ref{sec:intro} and the objectives that we aim to solve:
	\begin{equation*}
		\begin{split}
			M_k^N &\coloneqq \lbrace m_{k, \cu{n}} \colon \abs{\cu{n}} \leq 2 \, N - 1 \rbrace, \\
			m_{k, \cu{n}} &\coloneqq \expec{X_k^\cu{n} \cond Y_{1:k}} \coloneqq \int x^\cu{n} \diff \PP_{X_k \cond Y_{1:k}}(x),
		\end{split}
	\end{equation*}
	where $\PP_{X_k \cond Y_{1:k}}$ is the filtering probability measure at time $t_k$. In this section, we apply the moment quadrature introduced in Section~\ref{sec:nd-quadrature} to construct an approximation 
	\begin{equation*}
		\begin{split}
			\widehat{M}_k^N &\coloneqq \lbrace \widehat{m}_{k, \cu{n}} \colon \abs{\cu{n}} \leq 2 \, N - 1 \rbrace, \quad m_{k, \cu{n}} \approx \widehat{m}_{k, \cu{n}}, 
		\end{split}
	\end{equation*}
	such that $\widehat{m}_{k, \cu{n}}$ converges to $m_{k, \cu{n}}$ as $N\to\infty$ for any step $k$. Moreover, this algorithm is recursive in time in the way that the approximation $\widehat{M}_k^N$ at $t_k$ is only given by the previous approximation $\widehat{M}_{k-1}^N$ at $t_{k-1}$ and the current measurement $Y_k$. The recursion starts from the moments $M_0^N$ of the initial $\PP_{X_0}$ which we know exactly. This filter is detailed as follows.
	
	Suppose that at any time $t_{k-1}$ we have the approximate (or exact) $\widehat{M}_{k-1}^N$ of the true moments in $M_{k-1}^N$. The approximation must be valid in the sense that the Gram matrix built by the approximate moments is positive definite. Let us denote the quadrature rules generated by Algorithm~\ref{alg:nd-moment-quadrature} based on $\widehat{M}_{k-1}^N$ as $\lbrace w_{k-1, \cu{n}}, \lambda_{k-1, \cu{n}} \rbrace_{\cu{n} \in \mathfrak{n}_{N, d}}$. Then, we can propagate $\widehat{M}_{k-1}^N$ through the SDE to approximate that of the measure $\PP_{X_k \cond Y_{1:k-1}}$ at time $t_k$. Specifically, by Chapman--Kolmogorov equation, the moments
	\begin{equation}
		\begin{split}
			\int x^\cu{n} \diff \PP_{X_k \cond Y_{1:k-1}}(x) &= \int \expec{X_k^\cu{n} \cond X_{k-1} = x} \diff \PP_{X_{k-1} \cond Y_{1:k-1}}(x) \\
			&\approx \sum_{\cu{q} \in \mathfrak{n}_{N, d}} w_{k-1, \cu{q}} \expec{X_k^\cu{n} \cond X_{k-1} = \lambda_{k-1, \cu{q}}} \coloneqq \vscriptalign{\overline{m}}_{k, \cu{n}},
			\label{equ:filter-pred}
		\end{split}
	\end{equation}
	for all $\cu{n}$ such that $\abs{\cu{n}} \leq 2 \, N - 1$, are approximated by $\vscriptalign{\overline{M}}_k^N \coloneqq \lbrace \vscriptalign{\overline{m}}_{k, \cu{n}} \colon \abs{\cu{n}} \leq 2 \, N - 1 \rbrace$. The conditional expectation $x \mapsto \expec{X_k^\cu{n} \cond X_{k-1} = x}$ in the integral is the key that characterises the transition of the moments in time. However, the conditional expectation for the SDE is often intractable, and we have to find an approximation so as to evaluate the quadrature nodes.
	
	One commonly used approximation is the Euler--Maruyama scheme (or other Gaussian-based approximations). By applying this scheme, we are approximating the conditional expectation $\expec{X_k^\cu{n} \cond X_{k-1} = x}$ by the $\cu{n}$-moment of a multivariate Normal random variable with mean $x + a(x) \, (t_k - t_{k-1})$ and covariance $b(x) \, b(x)^\trans (t_k - t_{k-1})$. The moment of such is analytically available by Isserlis' theorem, however, its computational cost is astronomical (e.g., it needs at least $(\abs{\cu{n}} - 1)!!$ summations over the covariance matrix elements). Although it is possible to significantly reduce the cost by Kan--Magnus method~\cite{Kan2008}, the computation required is still an expensive function of $\abs{\cu{n}}$. Apart from the computational difficulty, it is also hard to improve the approximation error. As an example, computing the moment based on higher-order It\^{o}--Taylor discretisations (e.g., Milstein) in closed form is possible only under limited conditions (e.g., diagonal dispersion $b$).
	
	Provided that the SDE coefficients are sufficiently smooth, we can represent the conditional expectation by a $J$ order Taylor moment expansion (TME)~\cite{ZhaoTME2020, Zhao2021Thesis}
	\begin{equation}
		\begin{split}
			\expec{g(X_k) \cond X_{k-1} = x} &= \sum_{j=0}^J (A^j g)(x) \, \frac{(t_k - t_{k-1})^j}{j!} + R(x, J, g, t_k, t_{k-1}),\\
			(Ag)(x) &\coloneqq \bigl(\nabla_x g(x)\bigr)^\trans a(x) + \frac{1}{2} \traceBig{b(x) \, b(x)^\trans \, \hessian_x g(x)}, 
			\label{equ:tme}
		\end{split}
	\end{equation}
	by choosing $g(x) = x^\cu{n}$, where $\nabla_x$ and $\hessian_x$ denote the gradient and Hessian, respectively, $A$ is the infinitesimal generator, and $R$ is the remainder. If the time interval $t_k - t_{k-1}$ is not significantly large, we can discard $R$ and use the truncated term as the approximation which converges as $J\to \infty$. The upside of this TME method is that the computation is scalable for approximating high-order moments. More precisely, unlike the Euler--Maruyama scheme, the number of calculations required in Equation~\eqref{equ:tme} is independent of $\cu{n}$ once we fix $J$. Although the equation has iterative gradients and Hessians, they are not difficult to implement with the help of automatic differentiations and Jacobian/Hessian-vector product solvers. The downside of this approach is that the approximate moments do not guarantee to form a positive definite Gram matrix, due to the truncation error. 
	
	We can then use $\vscriptalign{\overline{M}}_k^N$ to approximate that of $\PP_{X_k \cond Y_{1:k}}$. By Bayes' rule, we apply the change-of-measure
	\begin{equation}
		\frac{\diff \PP_{X_k \cond Y_{1:k}}}{\diff \PP_{X_k \cond Y_{1:k-1}}}(x) = \frac{p_{Y_k \cond X_k}(Y_k \cond x)}{\int p_{Y_k \cond X_k}(Y_k \cond z) \diff \PP_{X_k \cond Y_{1:k-1}}(z)},
		\label{equ:change-of-measure}
	\end{equation}
	thus, the $\cu{n}$-moment of $\PP_{X_k \cond Y_{1:k}}$ is given by
	\begin{equation}
		\begin{split}
			m_{k, \cu{n}} &= \frac{1}{h_k} \int x^\cu{n} \, p_{Y_k \cond X_k}(Y_k \cond x) \diff \PP_{X_k \cond Y_{1:k-1}}(x), \\
			h_k &\coloneqq \int p_{Y_k \cond X_k}(Y_k \cond x) \diff \PP_{X_k \cond Y_{1:k-1}}(x),
			\label{equ:filtering-moment-and-ell}
		\end{split}
	\end{equation}
	which we can again approximate by applying the quadrature method. Specifically, let $\lbrace \vscriptalign{\overline{w}}_{k, \cu{n}}, \vscriptalign{\overline{\lambda}}_{k, \cu{n}} \rbrace_{\cu{n} \in \mathfrak{n}_{N, d}}$ be the quadrature rules generated by $\vscriptalign{\overline{M}}_k^N$, then the approximate moment for $\PP_{X_k \cond Y_{1:k}}$ is 
	\begin{equation*}
		\begin{split}
			m_{k, \cu{n}} \approx \widehat{m}_{k, \cu{n}}&\coloneqq\sum_{\cu{q} \in \mathfrak{n}_{N, d}} \vscriptalign{\overline{w}}_{k, \cu{q}} \, (\vscriptalign{\overline{\lambda}}_{k, \cu{q}})^\cu{n} \, p_{Y_k \cond X_k}(Y_k \cond \vscriptalign{\overline{\lambda}}_{k, \cu{q}}) \, / \, \widehat{h}^N_k, \\
			\widehat{h}_k^N &\coloneqq \sum_{\cu{n} \in \mathfrak{n}_{N, d}} \vscriptalign{\overline{w}}_{k, \cu{n}} \, p_{Y_k \cond X_k}(Y_k \cond \vscriptalign{\overline{\lambda}}_{k, \cu{n}}).
		\end{split}
	\end{equation*}
	With the approximate moments $\widehat{M}^N_k$, we can then compute the next $\vscriptalign{\widehat{M}}^N_{k+1}$, and so forth for any step $k$ by repeating this process. 
	
	Observe that $h_k = p_{Y_{k} \cond Y_{1:k-1}}(Y_k \cond Y_{1:k-1})$. Hence, the filtering routine simultaneously enables an approximation to the negative log-likelihood
	\begin{equation}
		\ell(Y_{1:k}) \coloneqq -\log p_{Y_{1:k}}(Y_{1:k}) = -\sum_{j=1}^k \log h_k \approx -\sum_{j=1}^k \log \widehat{h}_{j}^N,
		\label{equ:neg-ll}
	\end{equation}
	which we can use to estimate parameters in the filtering model by maximum likelihood. Moreover, the likelihood approximated by the filter is differentiable with respect to the model parameters. Hence, it is straightforward to leverage efficient gradient-based optimisation algorithms by means of automatic differentiations.
	
	We summarise the filter with moment representations in the following algorithm. 	
	\begin{algorithm2e}[h]
		\SetAlgoLined
		\DontPrintSemicolon
		\KwInputs{Order $N$, measurements $Y_{1:T}$, and initial moments $M_0^N$}
		\KwOutputs{Moments $\widehat{M}_1^N, \widehat{M}_2^N, \ldots, \widehat{M}_T^N$ and negative log-likelihood $\widehat{\ell}$}
		$\widehat{M}^N_0$ = $M^N_0$\;
		$\widehat{\ell} = 0$\;
		\For{$k=1$ \KwTo $T$}{%
			\tcp{Prediction step}
			$\lbrace w_{k-1, \cu{n}}, \lambda_{k-1, \cu{n}} \rbrace_{\cu{n} \in \mathfrak{n}_{N, d}}=$ \FnMQuad{$\widehat{M}^N_{k-1}$}\;
			\For(\tcp*[f]{In parallel}){$\cu{n}\colon \abs{\cu{n}} \leq 2 \, N - 1$}{%
				$\vscriptalign{\overline{m}}_{k, \cu{n}} = \sum_{\cu{q} \in \mathfrak{n}_{N, d}} w_{k-1, \cu{q}} \expec{X_k^\cu{n} \cond X_{k-1} = \lambda_{k-1, \cu{q}}} $\;
			}
			$\vscriptalign{\overline{M}}^N_{k}= \lbrace \vscriptalign{\overline{m}}_{k, \cu{n}} \colon \abs{\cu{n}} \leq 2 \, N - 1 \rbrace$\;
			\tcp{Update step}
			$\lbrace \overline{w}_{k, \cu{n}}, \overline{\lambda}_{k, \cu{n}} \rbrace_{n \in \mathfrak{n}_{N, d}}=$ \FnMQuad{$\vscriptalign{\overline{M}}^N_{k}$}\;
			$\widehat{h}_k^N = \sum_{\cu{n} \in \mathfrak{n}_{N, d}} \vscriptalign{\overline{w}}_{k, \cu{n}} \, p_{Y_k \cond X_k}(Y_k \cond \vscriptalign{\overline{\lambda}}_{k, \cu{n}})$\;
			\For(\tcp*[f]{In parallel}){$\cu{n}\colon \abs{\cu{n}} \leq 2 \, N - 1$}{%
				$\widehat{m}_{k, \cu{n}} = \sum_{\cu{q} \in \mathfrak{n}_{N, d}} \vscriptalign{\overline{w}}_{k, \cu{q}} \, (\vscriptalign{\overline{\lambda}}_{k, \cu{q}})^\cu{n} \, p_{Y_k \cond X_k}(Y_k \cond \vscriptalign{\overline{\lambda}}_{k, \cu{q}}) \, / \, \widehat{h}_k^N$\;
			}
			$\widehat{\ell} = \widehat{\ell} - \log \widehat{h}_k^N$\;
			$\widehat{M}^N_{k}= \lbrace \widehat{m}_{k, \cu{n}} \colon \abs{\cu{n}} \leq 2 \, N - 1 \rbrace$\;
		}
		\caption{Moment filter}
		\label{alg:moment-filter}
	\end{algorithm2e}

	\subsection{Computational complexity}
	\label{sec:filter-complexity}
	As shown in Algorithm~\ref{alg:moment-filter}, the moment filter is a sequential algorithm in time, hence, the time complexity is linear in the number of measurements $T$. At each filtering step, the computation cost is dominated by either the Cholesky decomposition of the Gram matrix (of size $S\times S$), or summing over the quadrature evaluations (of length $S^d$) depending on the actual implementation. If the summation is implemented schoolbook sequentially, then the summation complexity is $O(S^d)$ which is greater than that of the Cholesky decomposition if $d > 3$. On the other hand, if the summation is implemented in parallel, then the cost of the Cholesky decomposition dominates, which is $O(S^3)$. At every step, the filter computes the summation and Cholesky decomposition three and two times, respectively. Overall, Algorithm~\ref{alg:moment-filter} has time complexity
	\begin{itemize}
		\item $O(2 \, T \, S^3)$, if the summation is implemented in parallel, 
		\item $O(2 \, T \, S^3)$, if the summation is implemented sequentially and $d \leq 3$, or 
		\item $O(3 \, T \, S^d)$, if the summation is implemented sequentially and $d > 3$.
	\end{itemize}
	Recall that $S = \binom{N - 1 + d}{N - 1}$, and that $S \sim O(N^d \, / d!)$ if we fix the dimension $d$. Therefore, the time complexity of the filter is polynomial in the order $N$, and the degree of the polynomial is determined by the state dimension and the actual implementation of the filter.
	
	The time complexity of the filter is not significantly impacted by the measurement variable dimension $d_y$. At each step $k$, the filter evaluates the measurement density function $x \mapsto p_{Y_k \cond X_k}(Y_k \cond x)$ by the quadrature nodes. The complexity of each evaluation depends on $d_y$, for example, $O(d_y^3)$ for multivariate Normal distributions. But on the other hand, these evaluations are independent which can be done in parallel, hence, the complexity is not multiplied by the number of quadrature nodes. In practice, $d_y$ is far less than $S$ when the order $N$ is large. 
	
	\subsection{Convergence analysis}
	\label{sec:filter-convergence}
	We aim to show that the moment filter converges in both moments and distribution as $N\to\infty$ at every filtering step. This is intuitive, since the moment quadrature that we use is exact for polynomials of degree equal to or less than $2 \, N - 1$. If additionally the limiting distribution is determined by its moments, then the approximation converges in distribution too followed by the method of moments. 
	
	Recall that the filter is a recursive chain of approximations. Therefore, if we can prove that for any fixed step $k-1$ the convergence of  $\widehat{M}^N_{k-1}$ to $\PP_{X_{k-1} \cond Y_{1:k-1}}$ implies the convergence of $\widehat{M}^N_{k}$ to $\PP_{X_{k} \cond Y_{1:k}}$, then by mathematical induction the filter converges at every step as long as the initial moments converges to $\PP_{X_0}$. To keep the results clean, let us use shorthand $I_\mu(f) \coloneqq \int f \diff \mu$, and denote $\delta_x$ the Dirac measure at any point $x$. If a measure $\mu$ converges weakly and in moments to another measure $\nu$, then we denote $\mu\xrightarrow{\mathrm{w.m.}}\nu$. All the probability measures/distributions in this section operate on the same canonical space $(\R^d, \mathcal{B}(\R^d))$, where $\mathcal{B}$ stands for the Borel sigma-algebra. Then, we make the following lemma that backbones the filtering convergence for any fixed step. 
	\begin{lemma}
		\label{lemma:convergence-measure}
		Let $\mu$ and $\nu$ be two probability measures such that 1) $\nu$ is determined by its moments; 2) for any monomial $\eta_{\cu{n}}$, there is a bounded continuous, or polynomial function $g_{\eta_{\cu{n}}}$ that $I_\nu(\eta_{\cu{n}}) = I_{\mu}(g_{\eta_{\cu{n}}})$. Let $\widehat{\mu}_N$ be any finite measure that $\widehat{\mu}_N\xrightarrow{\mathrm{w.m.}}\mu$ as $N\to\infty$, and suppose that the Gram matrix generated by $\vscriptalign{\mathcal{M}}^N\coloneqq \lbrace I_{\widehat{\mu}_N}(g_{\eta_{\cu{n}}}) \colon \abs{\cu{n}} \leq 2 \, N - 1\rbrace$ is positive definite for all $N$. Then the measure $\widehat{\nu}_N \coloneqq \sum_{\cu{n} \in \mathfrak{n}_{N, d}} w_\cu{n} \, \delta_{\lambda_{\cu{n}}} \xrightarrow{\mathrm{w.m.}} \nu$ as $N\to\infty$, where $\lbrace w_\cu{n}, \lambda_{\cu{n}}\rbrace_{\cu{n} \in \mathfrak{n}_{N, d}}$ are the quadrature rules generated by $\mathcal{M}^N$.
	\end{lemma}
	\begin{proof}
		Since $\widehat{\mu}_N$ converges weakly and in moments to $\mu$, we have for any monomial $\eta_{\cu{n}}$ the expectation $I_{\widehat{\mu}_N}(g_{\eta_{\cu{n}}}) \to I_{\mu}(g_{\eta_{\cu{n}}})$ as $N\to\infty$. Therefore, the approximate moment $I_{\widehat{\mu}_N}(g_{\eta_{\cu{n}}})$ converges to the true moment $I_\nu(\eta_\cu{n})$ of $\nu$ for all $\cu{n}$. By the definition of the quadrature method, the moment of $\widehat{\nu}_N$ is
		\begin{equation*}
			I_{\widehat{\nu}_N}(\eta_\cu{n}) = 
			\begin{cases}
				I_{\widehat{\mu}_N}(g_{\eta_{\cu{n}}}) \in \mathcal{M}^N, & \abs{\cu{n}} \leq 2 \, N - 1, \\
				\sum_{\cu{q} \in \mathfrak{n}_{N, d}} w_\cu{q} \, \eta_\cu{n}(\lambda_{\cu{q}}) < \infty, & \abs{\cu{n}} > 2 \, N - 1,
			\end{cases}
		\end{equation*}
		which is finite for all $\cu{n}$ and $N$. Hence, for any fixed $\cu{n}$, we can always find a large enough $N$ such that $I_{\widehat{\nu}_N}(\eta_\cu{n}) = I_{\widehat{\mu}_N}(g_{\eta_{\cu{n}}})$, and that $I_{\widehat{\nu}_N}(\eta_\cu{n}) \to I_\nu(\eta_\cu{n})$ as $N\to\infty$. This proves that $\widehat{\nu}_N$ converges in moment to $\nu$. It is then followed by Fr\'{e}chet--Shohat theorem~\cite[pp. 540]{Frechet1931} that the measure $\widehat{\nu}_N$ converges weakly to $\nu$ as well.
	\end{proof}

	The convergence of the moment filter is then a result of iterative applications of Lemma~\ref{lemma:convergence-measure}. This result is concluded in Proposition~\ref{prop:convergence-filter} under the following model assumptions.
	
	\begin{assumption}
		\label{assump:moment-determint-filter}
		The distributions $\PP_{X_0}$, $\PP_{X_k \cond Y_{1:k}}$, and $\PP_{X_k \cond Y_{1:k-1}}$ for $k\geq 1$ are determined by their moments. 
	\end{assumption}
	\begin{assumption}
		\label{assump:init}
		The initial approximation $\widehat{\PP}^N_{X_0}\coloneqq \sum_{\cu{n} \in \mathfrak{n}_{N, d}} w_{k, \cu{n}} \, \delta_{\lambda_{k, \cu{n}}} \xrightarrow{\mathrm{w.m.}} \PP_{X_0}$.
	\end{assumption}
	\begin{assumption}
		\label{assump:integrand-regularity}
		For every $\cu{n}$, the function $x\mapsto \expec{X_k^\cu{n} \cond X_{k-1} = x}$ is bounded and continuous, or polynomial. Almost surely for every $\cu{n}$ and $k$, the functions $x \mapsto p_{Y_k \cond X_k}(\cdot\cond x)$ and $x \mapsto x^\cu{n} \, p_{Y_k \cond X_k}(\cdot \cond x)$ are bounded and continuous, or polynomial. 
	\end{assumption}
	\begin{assumption}
		\label{assump:stable-filter}
		For all $N\geq 1$ and $k\geq 0$, the Gram matrices generated by the moments in $\widehat{M}_k^N$ and $\vscriptalign{\overline{M}}^N_k$ are positive definite.
	\end{assumption}
	\begin{proposition}
		\label{prop:convergence-filter}
		Suppose that Assumptions~\ref{assump:moment-determint-filter} to~\ref{assump:stable-filter} are satisfied. Then almost surely for every $k\geq 0$, 
		\begin{equation}
			\widehat{\PP}^N_{X_k \cond Y_{1:k}} \xrightarrow{\mathrm{w.m.}} \PP_{X_k \cond Y_{1:k}},
			\label{equ:filter-convergent}
		\end{equation}
		as $N\to\infty$, where $\widehat{\PP}^N_{X_k \cond Y_{1:k}}\coloneqq \sum_{\cu{n} \in \mathfrak{n}_{N, d}} w_{k, \cu{n}} \, \delta_{\lambda_{k, \cu{n}}}$ is the discrete measure given by the moments $\widehat{M}_k^N$. Moreover, the approximate likelihood $\widehat{h}^N_k \to h_k$ almost surely.
	\end{proposition}
	\begin{proof}
		Assumption~\ref{assump:init} implies that Equation~\eqref{equ:filter-convergent} holds for $k=0$ by definition. Suppose that $\widehat{\PP}^N_{X_{k-1} \cond Y_{1:k-1}} \xrightarrow{\mathrm{w.m.}} \PP_{X_{k-1} \cond Y_{1:k-1}}$ for any $k > 1$. Then by Equation~\eqref{equ:filter-pred} we have the relation
		\begin{equation*}
			I_{\PP_{X_k \cond Y_{1:k-1}}}(\eta_{\cu{n}}) = I_{\PP_{X_{k-1} \cond Y_{1:k-1}}}(g_{\eta_\cu{n}}),
		\end{equation*}
		by taking $g_{\eta_\cu{n}}(x) = \expec{X_k^\cu{n} \cond X_{k-1} = x}$. 
		Hence, Lemma~\ref{lemma:convergence-measure} concludes that the measure $\sum_{\cu{n} \in \mathfrak{n}_{N, d}} \overline{w}_{k, \cu{n}} \, \overline{\delta}_{\lambda_{k, \cu{n}}} \xrightarrow{\mathrm{w.m.}} \PP_{X_k \cond Y_{1:k}}$, and this holds almost surely. It follows that $\widehat{h}^N_k \to h_k$ almost surely. Next, Equation~\eqref{equ:change-of-measure} similarly indicates the relation $I_{\PP_{X_k \cond Y_{1:k}}}(\eta_{\cu{n}}) = I_{\PP_{X_k \cond Y_{1:k-1}}}(g_{\eta_\cu{n}})$ by taking $g_{\eta_\cu{n}}(x) = x^\cu{n} \, p_{Y_k \cond X_k}(Y_k \cond x) \, / \, h_k$. Hence, by applying Lemma~\ref{lemma:convergence-measure} again, we have $\widehat{\PP}^N_{X_k \cond Y_{1:k}} \xrightarrow{\mathrm{w.m.}} \PP_{X_k \cond Y_{1:k}}$. The statement of the proposition is then concluded by induction.
	\end{proof}

	The proposition above shows that the moment filter is asymptotically convergent under a few assumptions on the system. Assumptions~\ref{assump:moment-determint-filter} and~\ref{assump:init} are not restrictive, since there are a large class of distributions determined by moments (see sufficient conditions in e.g., \cite[Chap. 15]{AchimKlenke2014}). Assumption~\ref{assump:integrand-regularity} asks the integrands used in the moment filter be either bounded continuous or polynomial, since the convergence in both distribution and moments directly implies the convergence of such expectations. A trivial example that satisfies this assumption is any linear Gaussian system. It is possible to generalise Assumption~\ref{assump:integrand-regularity} for analytical functions as well, if the dominated convergence theorem applies for every order power series expansion of the integrand.
	
	Assumption~\ref{assump:stable-filter} can be ambiguous, since it does not explicitly clarify what types of systems can guarantee the moment matrices be positive definite. But on the other hand, it is fundamentally hard for any quadrature method to ensure the approximate moments be jointly valid (e.g., Monte Carlo). One solution to this problem is to introduce another approximation on top of the system. For instance, we can approximate the transition $\expec{X_k^\cu{n} \cond X_{k-1}=x}$ by Euler--Maruyama, so that each quadrature evaluation gives a Gaussian moment which is always valid by definition. Another solution, but numerically, is to use LDL decomposition instead of Cholesky, and then clip the non-positive diagonal elements to small epsilons~\cite{Cheng1998}. This amounts to a moment-matrix completion by finding a new set of valid moments that are nearest (in Frobenius norm) to the approximate moments.

	\begin{remark}
		\label{remark:quadrature-independence}
		It is important to remark that Lemma~\ref{lemma:convergence-measure} is independent of our quadrature methods introduced in Section~\ref{sec:quadrature}. Specifically, the lemma holds for any moment-based quadrature, as long as it gives finite approximations and exact integrations for monomials of a degree determined by $N$. As a consequence, the convergence of the moment filter does not break if we replace the moment-quadrature with any such.
	\end{remark}

	\begin{figure}[t!]
		\centering
		\includegraphics[width=\linewidth]{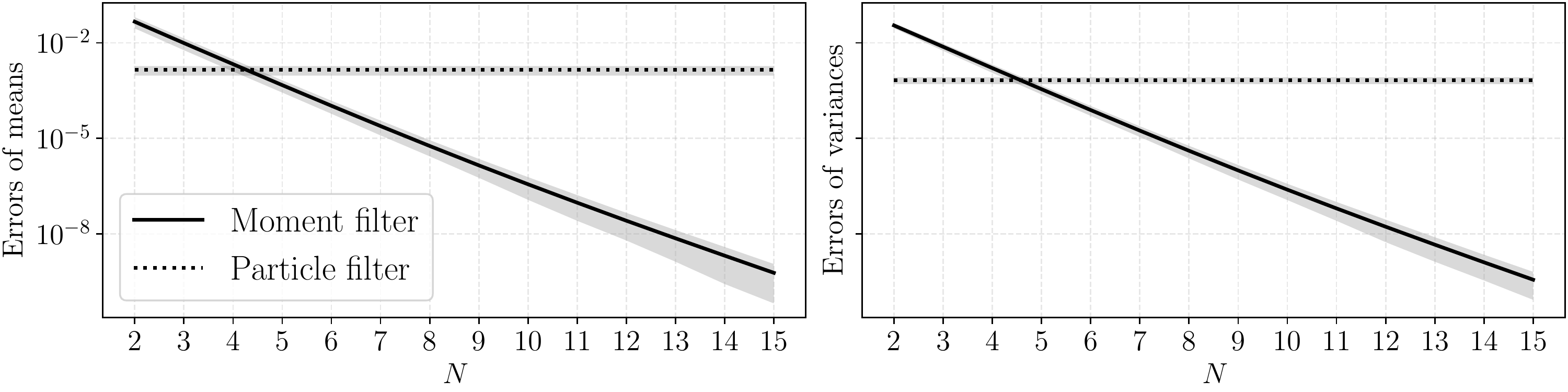}
		\includegraphics[width=\linewidth]{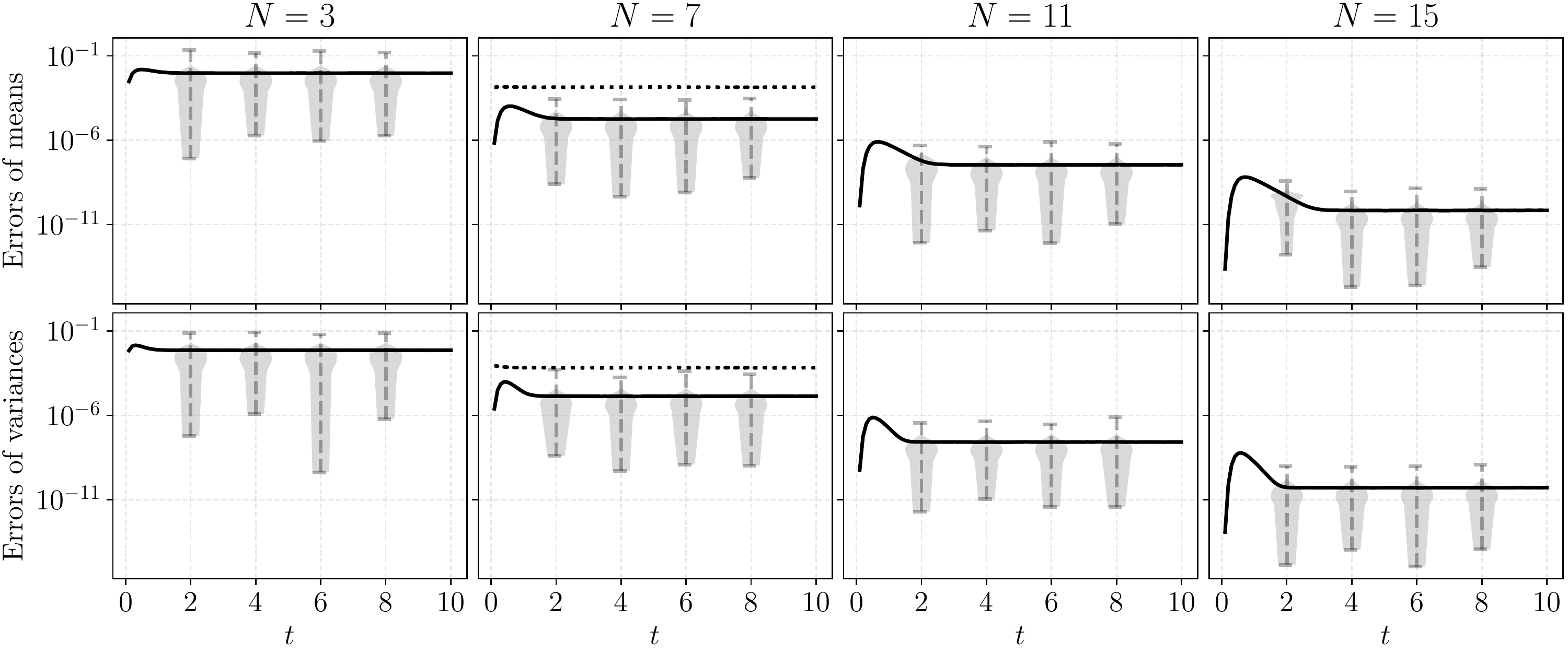}
		\caption{The filtering absolute errors (in log scale) for the model in Equation~\eqref{equ:model-lg}. The two figures in the first row show the errors as functions of the order $N$, where at each $N$ we average the errors over time, and then we plot the mean and two-standard deviation (i.e., the shaded area) from the MC simulations. The eight figures in the bottom two rows show the errors as functions of the time $t$, where we fix four $N$. The shaded violin plots show the distribution of the MC runs, with the dashed whiskers showing the extrema, and the solid black lines showing the means. The dotted lines show the mean errors of the particle filter for comparison.}
		\label{fig:convergence-lg}
	\end{figure}

	\section{Experiments}
	\label{sec:experiments}
	In this section, we conduct four experiments to numerically show the convergence of the moment filter, and to compare the performance against other commonly used filters. For numerical stability, we use central moments instead of raw moments in all the experiments (see, Remark~\ref{remark:modified-moments}). For approximating the conditional expectation in Equation~\eqref{equ:filter-pred}, we use the Taylor moment expansion of order three. All the experiments are implemented in the JAX of Python, and the implementations are published at \url{https://github.com/zgbkdlm/mfs} for reproducibility. \looseness=-1
	
	\subsection{Numerical convergence}
	\label{sec:experiment-convergence}
	To numerically show that the moment filter is convergent as $N\to \infty$, we test the filter on a linear Gaussian model
	\begin{equation}
		\begin{split}
			\diff X(t) &= -\frac{1}{\ell} \, X(t) \diff t +  \sqrt{\frac{2 \, \sigma^2}{\ell}} \diff W(t), \quad X(0) \sim \mathcal{N}(0, \sigma^2),\\
			Y_k &= X_k + \xi_k, \qquad\qquad\qquad\qquad\quad\quad\,\,\,\, \xi_k \sim \mathcal{N}(0, 1),
			\label{equ:model-lg}
		\end{split}
	\end{equation}
	where we fix the SDE parameters to $\ell = 1$ and $\sigma=0.5$. The reason for using this model is that we can exactly compute the true filtering distribution by a Kalman filter. To test the convergence statistically, we conduct 10,000 independent Monte Carlo (MC) simulations. For each MC simulation, we generate 100 measurements $Y_1, Y_2, \ldots, Y_{100}$ at evenly placed times $t_1 = 0.1, t_2 = 0.2, \ldots, t_{100} = 100$, respectively, then we compute the absolute errors of the filtering means and variances of the filter. More specifically, the two absolute errors are $\absbig{\widehat{m}_{k, 1} - \expec{X_k \cond Y_{1:k}}}$ and $\absbig{\widehat{v}_k - \varr{X_k \cond Y_{1:k}}}$ for $k=1,2,\ldots, 100$, where $\widehat{m}_{k, 1}$ and $\widehat{v}_k$ are the approximate filtering mean and variance, respectively.
	
	Furthermore, we apply a standard particle filter with 100,000 particles for comparison of the convergence. We use the stratified resampling at every step, and use the variance-optimal distribution~\cite[Thm. 10.1]{Chopin2020} as the proposal which is available in closed form for this model. 

	The results are shown in Figure~\ref{fig:convergence-lg}. From the first row of the figure, we see that the moment filtering error decreases as $N$ increases. Moreover, the convergence speed is numerically almost linear in the log scale, implying that the actual convergence speed is possibly a high degree polynomial of $N$. At $N=5$, the error of the moment filter starts to be better than that of the particle filter. On top of that, the moment filter with $N=5$ requires only 5 quadrature nodes and 10 moments, while the particle filter has 100,000 particles. When $N=15$, the convergence of the moment filter outperforms the particle filter out of a number of orders of magnitudes, while the actual running time of the moment filter is still faster than that of the particle filter.
	
	The bottom two rows in Figure~\ref{fig:convergence-lg} show the mean absolute errors as functions of time for a few fixed $N$. We see that at the initial time the errors are small, and then the errors increase as $t$ increases. This is true, since we know the exact moments of the initial random variable, and the moment filter by definition accumulates the filtering errors in time. However, we also observe that the errors shortly stop to increase (e.g., at $t\approx1$) and stabilise at certain levels. This suggests that the moment filter may have a bounded stability property which is worth investigating in future works.
	
	\subsection{Bene\v{s}--Bernoulli}
	\label{sec:benes-bernoulli}
	We next test the convergence of the moment filter with a non-linear SDE and Bernoulli binary measurements:
	\begin{equation}
		\begin{split}
			\diff X(t) &= \tanh(X(t)) \diff t + \diff W(t), \quad X(0) \sim \frac{1}{2}\bigl( \mathcal{N}(-0.5, 0.05) + \mathcal{N}(0.5, 0.05)\bigr), \\ 
			Y_k \cond X_k &\sim \mathrm{Bernoulli}\biggl(\frac{1}{1 + \exp(-X_k^3 \, / \, 5)}\biggr),
			\label{equ:benes-bernoulli}
		\end{split}
	\end{equation}
	whose filtering distribution is multimodal. For this model, it is hard to derive the exact solution analytically, but since the state is unidimensional, we can numerically compute the true solution by brute force up to machine precision. Specifically, we numerically compute the filtering PDFs (see, \cite[Algorithm 10.15]{Sarkka2019}) and their characteristic functions by trapezoidal rules at 2,000 spatial grids over a finite horizon where the true solution lies in. 
	
	We compare the moment filter to a bootstrap particle filter (10,000 samples with stratified resampling) and a Gaussian filter with Gauss--Hermite quadrature (of order 11). It is interesting to remark that the extended Kalman filters by definition do not work for this model, because the non-linear function in the Bernoulli parameter gives zero Kalman gain at the origin. Hence we do not compare to extended Kalman filters.
	
	We compute the approximation errors for the characteristic function of the filtering distribution. Denote $\varphi_k(z) \coloneqq \expec{\exp(\imag \, z \, X_k) \cond Y_{1:k}}$ as the true characteristic function, and $\widehat{\varphi}_k$ as the approximate. Then the error metric we use is $\sup_{z\in [-\gamma, \gamma]}\abs{\varphi_k(z) - \widehat{\varphi}_k(z)}$, where $\gamma=2$. For the moment filter and particle filter, we use its moments and samples, respectively to compute their approximate $\widehat{\varphi}_k$. For the Gauss--Hermite filter, its approximate characteristic function is given by that of the Normal distribution. To average the errors, we use 1,000 independent MC simulations, and in each simulation, we generate 100 measurements at evenly placed times $t_1 = 0.01, t_2 = 0.02, \ldots, t_{100} = 1$. Alongside the errors, we also present the running times of the filters. The running times are computed on a personal computer with Intel i9-10900K CPU. 
	
	\begin{figure}[t!]
		\centering
		\includegraphics[width=\linewidth]{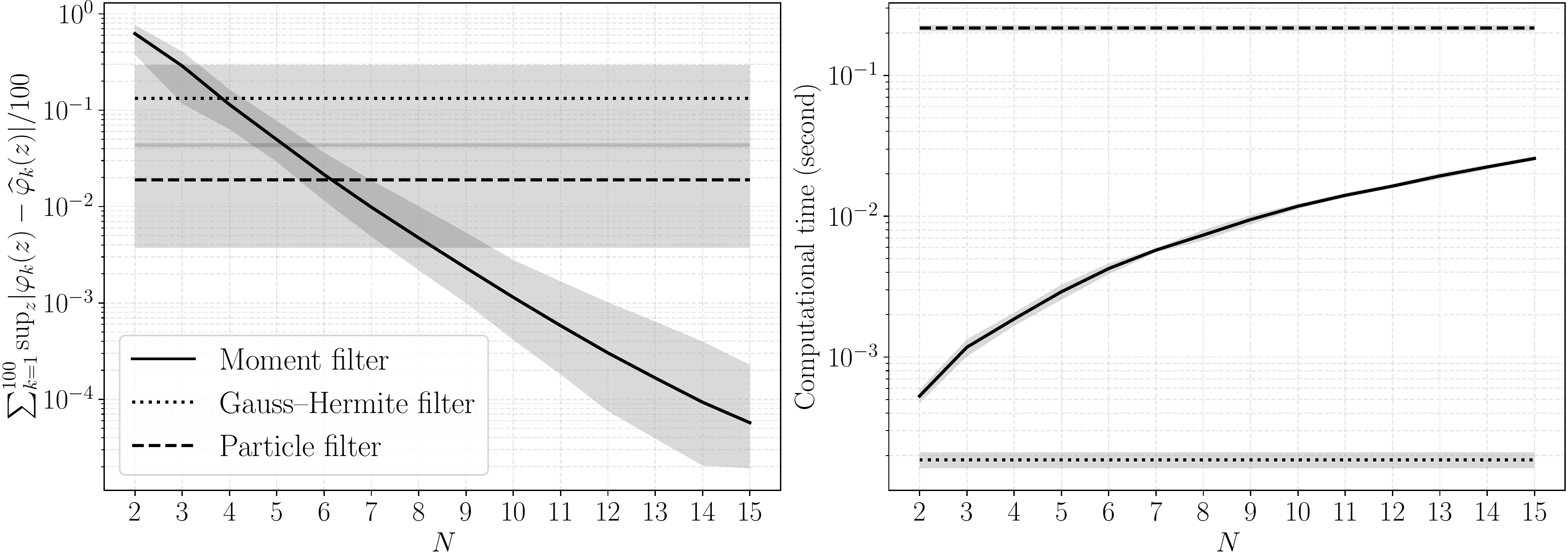}
		\caption{The estimation errors (left) and running times (right) of the filters for the Bene\v{s}--Bernoulli model in Equation~\eqref{equ:benes-bernoulli}. The shaded area plots the 0.95 quantile region computed from the 1,000 MC simulations.}
		\label{fig:benes-bernoulli-err}
	\end{figure}

	\begin{figure}[t!]
		\centering
		\includegraphics[width=\linewidth]{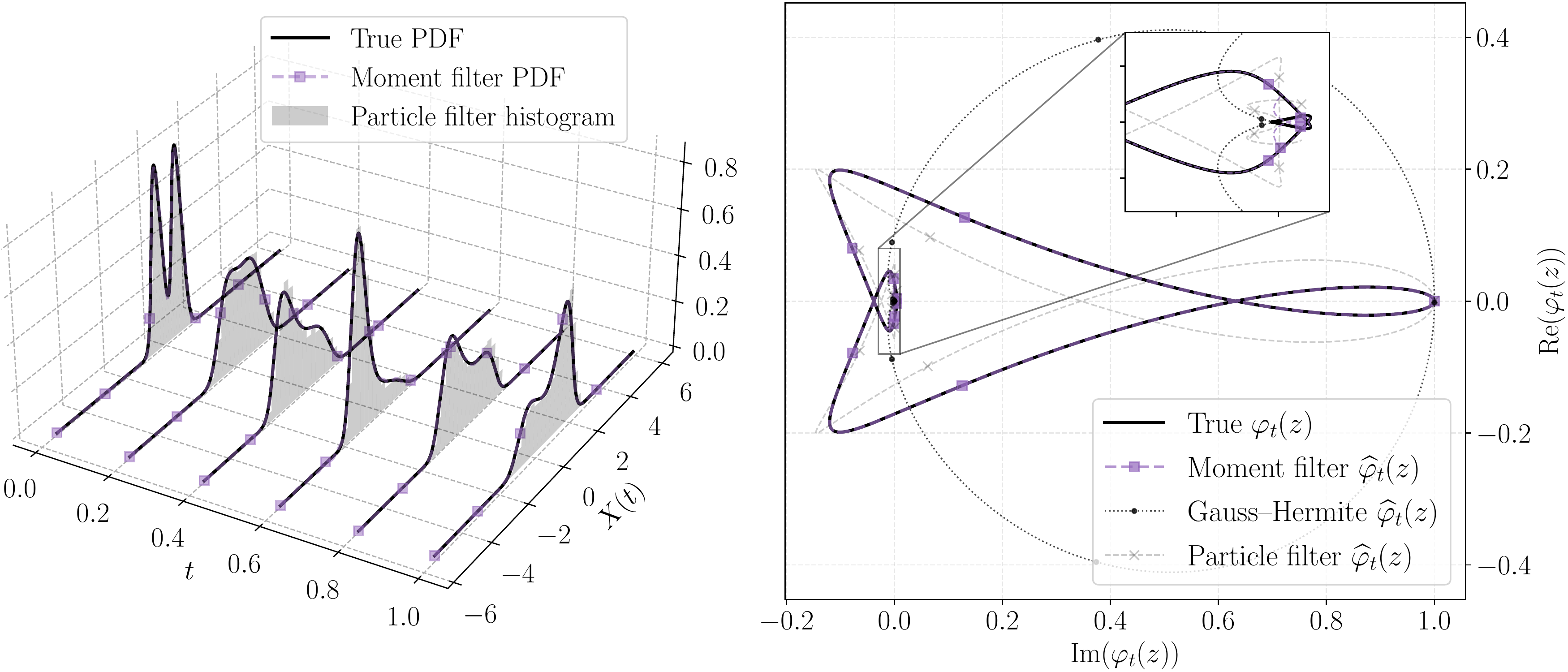}
		\caption{Demonstration of the filtering for the Bene\v{s}--Bernoulli model in Equation~\eqref{equ:benes-bernoulli}, with a fixed realisation of measurements. The left figure shows the evolution of the filtering PDFs. The right figure shows the characteristic functions at $t=0.8$ with $z\in[-9, 9]$. The PDF of the moment filter is numerically computed by the Fourier transform of the estimated characteristic function. The moment filter in this figure uses $N=15$, and its results overlap with the truth visually. }
		\label{fig:benes-bernoulli-demo}
	\end{figure}

	The estimation errors are shown on the left side of Figure~\ref{fig:benes-bernoulli-err}. We see that error of the moment filter is large when using low order of moments $N=2$. However, the error decreases in a near-linear speed (in log-scale) as we increase $N$, similar to the results in Section~\ref{sec:experiment-convergence}. At around $N=4$ and $N=6$, the error of the moment filter starts to be better than the Gauss--Hermite filter and particle filter, respectively. Meanwhile, as we increase $N$, the deviation of the moment filter error (i.e., the shaded area) enlarges too, but the deviation is smaller than other filters for most $N$. 
	
	The running times of the filters are shown in the right side of Figure~\ref{fig:benes-bernoulli-err}. We see that the Gauss--Hermite filter and particle filter are the fastest and slowest, respectively, with the moment filter in between the two. As we increase $N$, the speed of the moment filter decreases, but the decreasing speed is sub-linear (in log-scale). This is true, because the time complexity of the unidimensional moment filter is cubic in $N$ (see, Section~\ref{sec:filter-complexity}). At $N=15$, the accuracy of the moment filter is significantly better than the particle filter, while the speed is around ten times faster too. 
	
	In Figure~\ref{fig:benes-bernoulli-demo}, we demonstrate one simulation from the MC runs, and then plot the filtering PDF and characteristic function estimates. On the left figure, we see that the true PDFs are significantly non-Gaussian, while the estimated PDFs follow well the true PDFs. However, it is hard to see the differences between the moment filter and particle filter estimates. Hence, on the right figure, we plot the estimates for the characteristic function at $t=0.8$ which shows the difference more clearly. We find from the figure that the moment filter estimate is the closest to the truth, even at the tail around $\abs{z}=9$. On the other hand, the Gauss--Hermite filter's estimate is largely off, since it uses the Gaussian approximation. The particle filter has a better estimate than Gauss--Hermite, but its estimate deviates significantly from the truth compared to that of the moment filter.
	
	It is worth noting that the quadrature method in Section~\ref{sec:quadrature} with $N = 2$ is equivalent to Gauss--Hermite (GH) of order two when $d=1$. However, this does not mean that the moment filter with $N=2$ is the same as a Gaussian filter with GH of order two. The GH filter approximates the filtering distribution by Gaussian, while the moment filter does not. The two filters differ in computing Equations~\eqref{equ:filter-pred} and~\eqref{equ:filtering-moment-and-ell}, hence, they do not give the same results in Figure~\ref{fig:benes-bernoulli-err} at $N=2$.
	
	\begin{figure}[t!]
		\centering
		\includegraphics[width=\linewidth]{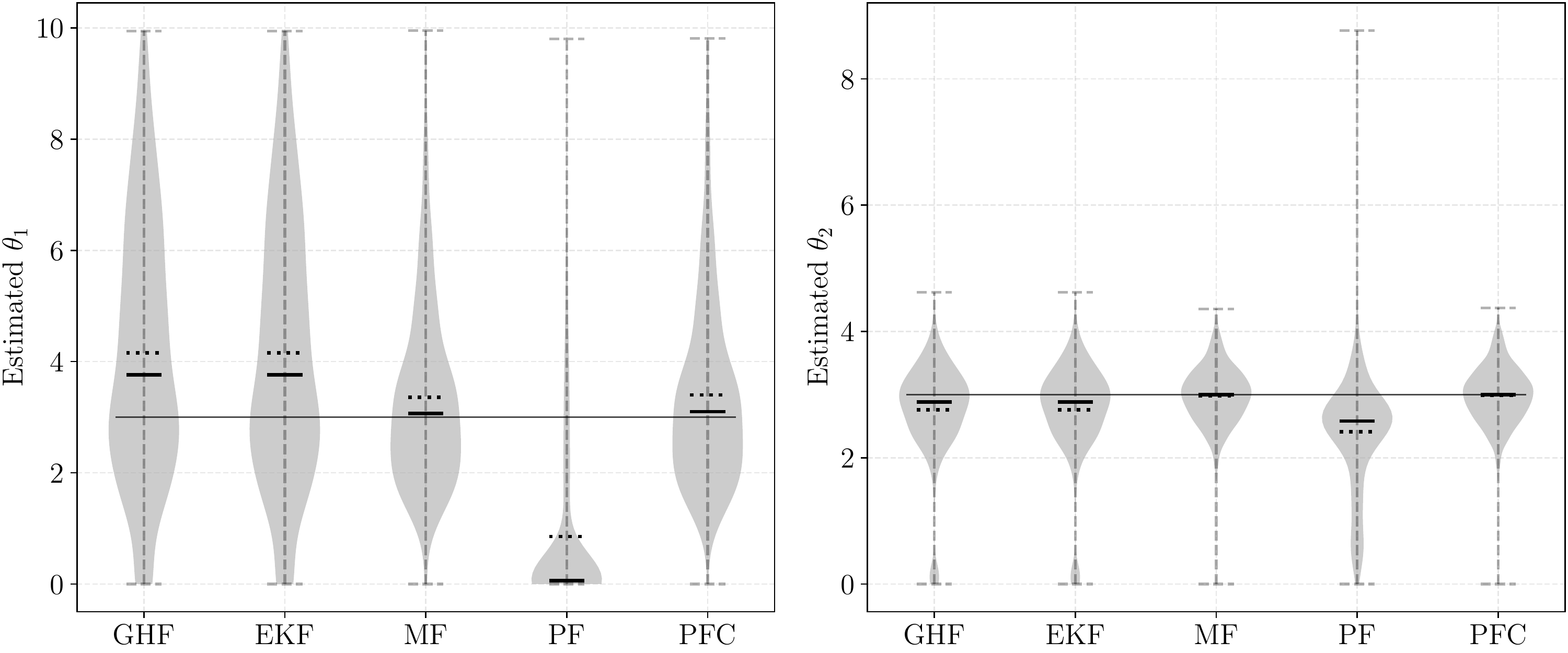}
		\caption{Estimates parameters for the model in Equation~\eqref{equ:well-poisson}. The grey violin plots show the distributions of the estimated parameters in all the MC runs, with the dashed whiskers representing the extrema. The solid and dotted black lines show the medians and means, respectively. The strike-through black lines represent the true parameter values. Moreover, in the 1,000 MC runs, GHF, EKF, MF, PF, and PFC have 105, 105, 34, 442, and 18 divergences, respectively.}
		\label{fig:param-est}
	\end{figure}
	
	\subsection{Parameter estimation}
	\label{sec:param-est}
	In this section we test the parameter estimation by minimising the negative log-likelihood produced by the filters (e.g., Equation~\eqref{equ:neg-ll}). The test model is
	\begin{equation}
		\begin{split}
			\diff X(t) &= X(t) \, \bigl( 1 - \theta_1 \, X(t)^2 \bigr) \diff t + \diff W(t), \\
			X(0) &\sim \frac{1}{2}\bigl( \mathcal{N}(-0.5, 0.05) + \mathcal{N}(0.5, 0.05)\bigr), \\ 
			Y_k \cond X_k &\sim \mathrm{Poisson}\Bigl( \log\bigl(1 + \exp(\theta_2 \, X_k)\bigr) \Bigr),
			\label{equ:well-poisson}
		\end{split}
	\end{equation}
	where the parameters are set to be $\theta_1=\theta_2=3$. To minimise the objectives, we use an L-BFGS-B optimiser with a positive bijection $x \mapsto \log(\exp(x) + 1)$ over the parameters which ensures positive estimates. The gradients with respect to the parameters are obtained by automatic differentiation. The initial values for the two parameters are uniformly set to be 0.1. 
	
	We conduct 1,000 independent MC runs, and in each MC run we generate 1,000 measurements at times $t_1 = 0.01, t_2 = 0.02, \ldots, t_{1000} = 10$. In addition, if the optimiser numerically diverges, or the estimated parameters exceed the threshold 10, then we mark the corresponding run as divergent. For comparison, we test the moment filter with $N=7$ against the Gauss--Hermite filter (GHF), the extended Kalman filter (EKF), and the bootstrap particle filter (PF). The configurations of these filters are the same as in Section~\ref{sec:benes-bernoulli}. However, the gradient produced by the PF with the standard resampling methods is biased. Hence, for a fair comparison, we use the continuous resampling method in~\cite{Corenflos2011} to compensate such biasedness. We abbreviate this particle filter with continuous resampling as PFC.
	
	The distribution of the estimated parameters are shown in Figure~\ref{fig:param-est}. We see that the estimates of the moment filter are evidently closest to the truth. In particular for $\theta_1$, the distribution of the moment filter estimates is more concentrated around the truth and is less tailed, compared to GHF and EKF. Moreover, for $\theta_2$, GHF and EKF have a few runs that give near-zero estimates, but, this problem does not appear for the moment filter. Among all the filters, PF is numerically the worst, as it does not give reasonable estimates for $\theta_1$, and the estimates for $\theta_2$ have large extrema and are heavily tailed. The performance of PFC is comparable to the moment filter, at the cost of using the continuous resampling. 
	
	In all the 1,000 MC runs, GHF, EKF, MF, PF, and PFC have 105, 105, 34, 442, and 18 divergences, respectively. This shows that the moment filter is relatively stable for this parameter estimation task compared to other filters. 
	
	\begin{figure}[t!]
		\centering
		\includegraphics[width=.7\linewidth]{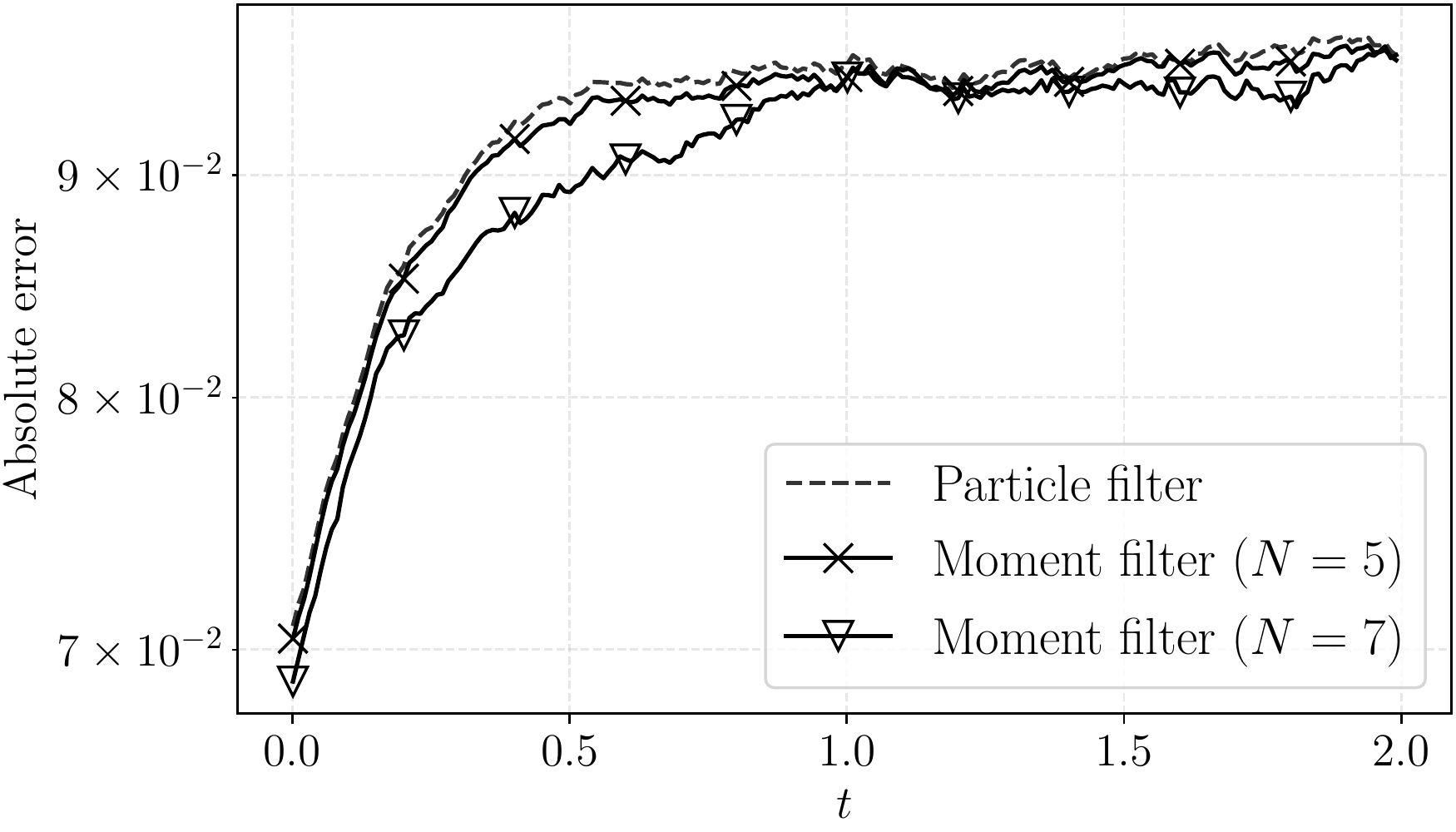}
		\caption{The estimation absolute errors for the prey-predator model in Equation~\eqref{equ:prey-predator}.}
		\label{fig:prey-predator}
	\end{figure}
	
	\subsection{Prey-predator}
	\label{sec:prey-predator}
	To test the performance of the moment filter for multidimensional systems, we consider the following prey-predator model
	\begin{equation}
		\begin{split}
			\diff X^{(1)}(t) &= X^{(1)}(t) \, \bigl(\alpha - \beta \, X^{(2)}(t)\bigr) \diff t + \sigma \, X^{(1)}(t)\diff W^{(1)}(t), \\
			\diff X^{(2)}(t) &= X^{(2)}(t) \, \bigl( \zeta \, X^{(1)}(t) - \gamma \bigr) \diff t + \sigma \, X^{(2)}(t)\diff W^{(2)}(t), \\
			Y_k \cond X_k &\sim \mathrm{Poisson}\Bigl( 1 \, / \, \bigl(1 + \exp\bigl(-\bigl(X^{(1)}_k\bigr)^3 + 1\bigr) \bigr) \Bigr), 
			\label{equ:prey-predator}
		\end{split}
	\end{equation}
	where $\alpha = \beta=\zeta=\gamma=4$, and $\sigma=0.1$. The prey-predator equation is commonly used by computational ecologist for modelling population growth, where $X^{(1)}$ and $X^{(2)}$ represent the populations of the preys and predators, respectively. To ensure the states be positive, we let the initial distribution of the SDE be a Normal with mean ones and small diagonal covariance scaled by $10^{-3}$. We simulate this model with 10,000 MC runs at 2,000 times $t_1 = 0.001, t_2 = 0.002, \ldots, t_{2000} = 2$. However, for this model it is hard to compute the true filtering solution. Hence, we instead compute the projection error $\expec{\norm{X(t_k) - \widehat{m}_{k, 1}}_1}$ at every time $t_k$, where $\widehat{m}_{k, 1}$ is the estimated filtering mean, and $\norm{\cdot}_1$ is the absolute norm. We use the same particle filter as in the previous section, for comparison.
	
	The results are shown in Figure~\ref{fig:prey-predator}. We see that the moment filter with $N=5$ is already slightly better than the particle filter. With $N=7$, the moment filter outperforms the particle filter. In average, the particle has absolute error $9.21\times 10^{-2}$, while the moment filter with $N=5$ and $N=7$ have error $9.15\times 10^{-2}$ and $9.03\times 10^{-2}$, respectively. However, we remark that in all the 10,000 MC runs, the moment filter numerically diverges 909 and 4,746 times for $N=5$ and $N=7$, respectively. The filter diverges because the Gram matrices generated by the moments are numerically not positive definite. As a contrast, all the particle filter runs are stable. This shows that although the moment filter is more performant than the particle filter, the moment filter can be numerically unstable depending on the model. 
	
	\section{Conclusions}
	\label{sec:conclusions}
	In this paper, we have developed a stochastic filter that represents the filtering distributions by approximate moments. Moreover, we proved that this filter is asymptotically exact, in the sense that the approximation converges to the truth in both distribution and moments, in the number of moments used. To make the moment filter computable, we have also developed a moment-based quadrature method based on the finite-matrix representation of multiplication operators. Our experiments showed that the moment filter numerically converges to the true solutions, and the filter outperforms a number of commonly used filters in terms of both computation time and estimation error. 
	
	We would also like to remark the limitation of the method for high-dimensional systems. As shown in Section~\ref{sec:filter-complexity}, the computational complexity of the moment filter does not scale well in the state dimension, due to the Cartesian-product construction of quadrature rules. But on the other hand, as we have mentioned in Remark~\ref{remark:quadrature-independence}, the convergence of the moment filter is independent of the used quadrature method. Hence, to reduce the computation, we can replace the moment quadrature with any that is more efficient, for example, by introducing a sparse version of Algorithm~\ref{alg:moment-filter}. It is reasonable to come up with a sparse routine, since from Figure~\ref{fig:nd-quadrature-rules} we evidently see plenty of quadrature weights that are too small to be useful. Another problem for high-dimensional systems is the numerical stability. As shown in the experiment in Section~\ref{sec:prey-predator}, the moment filter can numerically diverge due to non-positive-definite Gram matrices. This problem can be numerically solved by modified Cholesky decompositions to certain extent, but it is unclear whether this breaks the convergences. Therefore, developing a more efficient and numerical stable moment quadrature method is an important topic for future investigations.
	
	The convergence analysis in Section~\ref{sec:filter-convergence} proves that the moment filter converges, but it does not explicitly quantify the approximation error. For future works, it is interesting to measure how the error accumulates in time, and also to see if the error has a finite bound. Figure~\ref{fig:convergence-lg} numerically shows that the error has a finite bound for a specific model. 
	
	\section{Related works}
	\label{sec:related-works}
	The essence of the paper consists in representing the filtering distributions by a sequence of moments. On the special case of continuous-time measurements, this idea translates into projecting the Kushner--Stratonovich equation solution onto a finite-dimensional basis spanned by the moments. As an example, in~\cite[Equ. 15]{Kushner1967} and~\cite[Sec. 7]{Brigo1999} they show a system of differential equations of moments, so that we can approximately compute the moments by simulating the moment equations. These moment equations, however, have intractable expectations, hence, they in addition approximate the filtering distributions by, for instance, Gaussian~\cite{Kushner1967} or exponential families~\cite{Brigo1999}. Using such approximation makes their approximate moments consistent but does not guarantee the approximate moments converge to the truth. On the other hand, we can in principle solve these moment equations by the moment quadrature as well. This gives a convergent continuous-time filter, at the cost of losing the consistency using finite number of moments. As such, our work in this paper can be seen as a convergent projection filter for discrete-time measurements. 
	
	The work in~\cite{Luo2016} (cf.~\cite{Boutayeb1997}) similarly solves the filtering problem by moments, except that they apply Taylor expansions to solve the moment-quadrature problem and aim for the continuous-time filtering. However, as we have argued in Section~\ref{sec:quadrature}, applying Taylor expansions to solve the moment quadrature problem imposes strict conditions on the integrands and the underlying distribution, otherwise the quadrature does not converge. During the initial phase of working on our paper, we have experimented with applying Taylor expansions, but the resulting filters numerically diverge for all the models that we show in the experiment section. Moreover, using high-order Taylor expansions is computationally hard, even with the aid of automatic differentiation. The experiments in~\cite{Luo2016} have demonstrated their method with moments of orders up to three, while ours can be up to thirty handily. It is also unclear whether the filter in~\cite{Luo2016} is in theory convergent.
	
	The Edgeworth series are common tools to approximate probability density functions by moments, and they are used in the filtering context as well~\cite{Challa2000, Singer2008}. The idea is to find a reference density function, and then approximate the filtering density by a product of the reference density and a polynomial of moments. If the reference density is chosen as a Gaussian, then Equations~\eqref{equ:filter-pred} and~\eqref{equ:filtering-moment-and-ell} can be approximated by Gauss--Hermite rules and a modification of the integrands. Essentially, this quadrature amounts to an importance integration, where the reference distribution is the importance distribution. However, in the filtering applications, it is in fact the Gram--Charlier series in~\cite{Challa2000, Singer2008}, since the Edgeworth expansion asks to decompose a random variable into independent and identically distributed ones. The convergence of the Gram--Charlier series imposes strict conditions on the tail of the distribution which limits the application of this method. 
	
	In~\cite{Kloeden2017}, they apply the same unidimensional Gaussian quadrature as in Section~\ref{sec:quadrature} for approximating the distributions of SDE solutions. In a sense, solving the SDE is a special filtering problem but without the measurement variables. Hence, the work~\cite{Kloeden2017} is seen as a special case of our method by discarding the update step in Algorithm~\ref{alg:moment-filter} for unidimensional systems.
	
	There are a number of studies that extend the Gaussian-approximate filters by leveraging high-order moments. As an example, in~\cite[Sec. III.C]{Julier2002}, \cite{Tenne2003}, \cite[Sec. V.E]{Menegaz2015}, and~\cite{Ponomareva2013}, they modify the unscented transform to additionally make use of the skewness and kurtosis. However, the resulting filters are still under the hood of Gaussian approximations to the true distributions. The fourth-order moment quadrature method in~\cite{Easley2021} is not restricted to the Gaussian approximation. Specifically, they represent the moments by tensors, and then use tensor decompositions of the moments to compute the quadrature rules. In principle, the quadrature method by~\cite{Easley2021} can be applied for the filtering problem as well, but it is unclear how to systematically derive the quadrature rules for moments higher than the fourth order. Even if it is possible to go beyond the fourth moment, representing the moments by tensors and computing the tensor decompositions are memory-consuming and computationally demanding. 
	
	In short, compared to the existing works, our contributions are significant in terms of the convergence and computation. 
	
	\section*{Authors' contributions}
	Zheng Zhao came up with the idea of the paper, did all the experiments, and wrote the initial draft. Juha Sarmavuori developed the moment quadrature methods and demonstrated them in Matlab. Juha Sarmavuori proved the convergence of the moment quadrature, and Zheng Zhao proved the convergence of the moment filter.
	
	\section*{Acknowledgments}
	The authors would like to thank Adrien Corenflos for his technical suggestions on the convergence, moment matrix completion, and continuous resampling, as well as Sebastian Mair, Jens Sj\"{o}lund, and Muhammad F. Emzir for their valuable comments.
	
	\bibliographystyle{siamplain}
	\bibliography{refs}
\end{document}